\documentclass[pdflatex,sn-mathphys-num]{sn-jnl}

\usepackage{amsmath}
\usepackage{amsfonts}
\usepackage{amssymb}
\usepackage{stmaryrd}
\usepackage{bm}
\usepackage{amsthm}
\usepackage{mathrsfs}

\usepackage[adversary,operators,mm,primitives,keys,complexity,advantage,notions]{cryptocode}
\createpseudocodeblock{pcb}{center, boxed }{}{}{}

\usepackage{float}
\usepackage{placeins}
\usepackage{tikz}
\usepackage{graphicx}
\usepackage{xcolor}
\usepackage{textcomp}
\usepackage{todonotes}

\usepackage{multirow}
\usepackage{booktabs}
\usepackage{array}

\usepackage{algorithm}
\usepackage{algorithmicx}
\usepackage{algpseudocode}

\usepackage{url}
\usepackage{hyperref}

\usepackage{qtree}
\usepackage{comment}
\usepackage{manyfoot}
\usepackage[title]{appendix}
\usepackage{listings}
\newcommand{\vecrow}[1]{Vec_{row}\left(#1\right)}

\newcommand{\rightinverse}[1]{\bm{R}_{#1}}

\newcommand{\CG}[2]{{
  \begin{bmatrix} \textit{#1} \\ \textit{#2} \end{bmatrix}_q
}}

\newcommand\C{\mathcal{C}}
\newcommand\D{\mathcal{D}}

\newcommand\Rank{\mathsf{rank}}
\newcommand\oneto[1]{[1, #1]}
\newcommand\MSE{\mathsf{MSE\new{P}}}
\newcommand\IMSEP{\mathsf{IMSEP}}
\newcommand\MCE{\mathsf{MCE\new{P}}}
\newcommand\QMLE{\mathsf{QMLE}}
\newcommand\QSMLE{\mathsf{QSMLE}}

\newcommand\iQSMLE{\mathsf{inhQSMLE}}
\newcommand\hQSMLE{\mathsf{hQSMLE}}
\newcommand\F{\mathcal{F}}
\newcommand\Pcal{\mathcal{P}}
\newcommand\GL{\mathsf{GL}}
\newcommand\Fq{\mathbb{F}_q}

\newcommand\Fqkmn{\Fq^{k\times mn}}

\newcommand\qm{q^m}
\newcommand\Fqm{\mathbb{F}_{\qm}}

\newcommand\sampler{\overset{\;\$}{\longleftarrow}}

\newcommand\accept{\textsf{accept}}
\newcommand\reject{\textsf{reject}}

\newcommand\share[1]{ [\![#1]\!] }

\newcommand{\old}[1]{\textcolor{gray!40}{}}
\newcommand{\new}[1]{\textcolor{black}{#1}}

\newcommand{\REQUIRE}{\Require}
\newcommand{\ENSURE}{\Ensure}
\algnewcommand{\PUBLICVAL}{\item[\textbf{Public values:}]} 
\newcommand{\IF}{\If}
\newcommand{\ENDIF}{\EndIf}
\newcommand{\STATE}{\State}
\newcommand{\REPEAT}{\Repeat}
\newcommand{\UNTIL}{\Until}

\theoremstyle{thmstyleone}\newtheorem{theorem}{Theorem}\newtheorem{proposition}[theorem]{Proposition}\newtheorem{problem}{Problem}
\newtheorem{lemma}{Lemma}

\theoremstyle{thmstyletwo}\newtheorem{remark}{Remark}

\theoremstyle{thmstylethree}\newtheorem{definition}{Definition}

\begin{document}

\title[Matrix Subcode Equivalence]{The Matrix Subcode Equivalence problem and its application to signature with MPC-in-the-Head}

\author[1]{\fnm{Magali} \sur{Bardet}}\email{magali.bardet@univ-rouen.fr}

\author[1]{\fnm{Charles} \sur{Brion}}\email{charles.brion@univ-rouen.fr}

\author[2]{\fnm{Philippe} \sur{Gaborit}}\email{gaborit@unilim.fr}

\author[2]{\fnm{Mercedes} \sur{Haiech}}\email{mercedes.haiech@unilim.fr}

\author[2]{\fnm{Romaric} \sur{Neveu}}\email{romaric.neveu@unilim.fr}

\affil[1]{\orgdiv{LITIS}, \orgname{University of Rouen Normandie},\orgaddress{ \city{Rouen}, \country{France}}}

\affil[2]{\orgdiv{XLIM}, \orgname{University of Limoges}, \orgaddress{\city{Limoges}, \country{France}}}

\abstract{
  Nowadays, equivalence problems are widely used in cryptography, most notably to establish cryptosystems such as digital signatures, with MEDS, LESS, PERK as the most recent ones. However, in the context of matrix codes, only the code equivalence problem has been studied, while the subcode equivalence is well-defined in the Hamming metric. In this work, we introduce two new problems: the Matrix Subcode Equivalence Problem and the Inhomogeneous Matrix Subcode Problem, to which we apply the Multi-Party-Computation-in-the-Head (MPCitH) paradigm to build a signature scheme. These new problems, closely related to the Matrix Code Equivalence problem, ask to find an isometry given a code $C$ and a subcode $D$. Furthermore, we prove that the Matrix Subcode Equivalence problem reduces to the Hamming Subcode Equivalence problem, which is known to be NP-Complete, thus introducing the matrix code version of the Permuted Kernel Problem. We also adapt the combinatorial and algebraic algorithms for the Matrix Code Equivalence problem to the subcode case, and we analyze their complexities. We find with this analysis that the algorithms perform much worse than in the code equivalence case, which is the same as what happens in the Hamming metric. Finally, our analysis of the attacks allows us to take parameters much smaller than in the Matrix Code Equivalence case. Coupled with the effectiveness of \textit{Threshold-Computation-in-the-Head} or \textit{VOLE-in-the-Head}, we obtain a signature size of $\approx$ 4 800 Bytes, with a public key of $\approx$ 275 Bytes. We thus obtain a reasonable signature size, which brings diversity in the landscape of post-quantum signature schemes, by relying on a new hard problem. In particular, this new signature scheme performs better than SPHINCS+, with a smaller size of public key + signature. Our signature compares also well with other signature schemes: compared to MEDS, the signature is smaller, and we reduced the size of the sum of signature and public key by a factor close to 5. We also obtain a signature size that is almost half the size of the CROSS signature scheme.
}
\maketitle
\section{Introduction}
\textbf{Background on post-quantum cryptography.} Ever since Shor's algorithm in 1994 \cite{Shor} to solve the factorization problem, cryptosystems such as RSA or Diffie-Hellman are known to be broken by quantum computers. While the quantum computers are nowhere near close to breaking the current cryptosystems in use, the major and recent advances in this field pushed the NIST to start a post-quantum standardization call in 2017 \cite{nistcall2017}, to propose an alternative to schemes that would be broken by quantum computers. After  Kyber, Dilithium, Falcon, and SPHINCS+ were chosen for standardization, another call for digital signatures has been started in 2023 \cite{nistcall2023}, to obtain more diversity in the signature schemes. In fact, amongst the three signature schemes Dilithium, Falcon, and SPHINCS+, two of them rely on lattices, and the last one on hash functions. While these schemes are still secure to this day, it is highly beneficial in cryptography to have as much diversity as possible. For this new signature call, many propositions relying on the MPC-in-the-Head have been made (8 out of 40), on a large range of assumptions: code-based (Hamming or Rank) \cite{RYDE,SDitHNIST,Mirath}, Multivariate Quadratic, \cite{mqom} and even symmetric assumptions \cite{FAEST}. Other schemes, involving equivalence problems, have been submitted: LESS \cite{LESS}, MEDS \cite{MEDS2}, or ALTEQ \cite{ALTEQ}. However, these schemes do not involve the MPC-in-the-Head paradigm. Out of these 40 submissions, 14 have been selected for the second round, with 6 of them relying on the MPC-in-the-Head, which truly shows its efficiency. It is thus natural to wonder whether equivalence problems can be turned into efficient signature schemes by using the MPCitH paradigm.

\textbf{Background on code equivalence problems.} The code equivalence problem is a well-known problem in cryptography, which has been studied and investigated for several years (\cite{PR97, Leon, S00} and many others). In the Hamming metric, the problem asks to find an isometry $\pi$ such that two codes $\C$ and $\D$ are such that $\pi(\C) = \D$. Leon proposed an algorithm to solve it in 1982 \cite{Leon}, which relies on finding codewords of small weight, and then finding the isometry from these codewords. In another work, \cite{PR97} proved that it is unlikely to be NP-complete. There have been further studies, such as \cite{S00}, with the Support Splitting algorithm, which is polynomial for random codes, but is much less efficient for codes with a large hull, or \cite{beullens}, which also attacks the problem. The code equivalence problem in the Hamming metric led to the signature scheme LESS, which has been submitted to the NIST call for signature in 2023 and has been selected for the second round. This signature has been improved in \cite{CPS24}, to obtain a signature size around $2$ kB. The downside of such a small signature is a large public key, of at least $14$ kB. \new{Introduced in \cite{Berger03}} for matrix codes, code equivalence has also been studied, in \cite{CDG21,RST22,MEDS1}, and several others (\cite{RS24,KQT24} amongst others). Endowed with the rank metric, the problem asks to find two invertible matrices $\bm{A}$ and $\bm{B}$ such that $\D = \bm{A}\C\bm{B}$. Similarly to the Hamming metric, the signature scheme MEDS has been submitted to the NIST. The problem reduces to the equivalence problem in the Hamming metric, and the attacks consist of an adaptation of Leon's algorithm, reduction to other equivalence problems, or algebraic attacks.

Furthermore, there is another problem one can look at when dealing with code equivalence in Hamming metric: the subcode equivalence problem, in which one knows a code and a subcode instead of two codes. This problem is equivalent to the Permuted Kernel Problem (PKP), which led to several works, either cryptosystems or cryptanalysis: \cite{SH90,pkpdss,SBC24,BG22,PERK,BBGK24} and more. While well-known and studied in the Hamming metric, PKP and the subcode equivalence problem have not been adapted to the rank metric, and there are, to the best of our knowledge, no cryptosystems based on the permuted kernel problem for matrix codes.

\textbf{The MPCitH paradigm.} Zero-knowledge proofs are crucial elements in cryptography. They allow a prover to authenticate himself to a verifier, without revealing any information on a secret. These proofs can then be transformed into a signature scheme, making them a useful tool in cryptography. A first such proof in cryptography was the Fiat-Shamir Authentication scheme \cite{FS87}, which proved the knowledge of a square root in $\mathbb{Z}_N$, with a cheating probability of $1/2$ per round. Then, in post-quantum cryptography, and more precisely, code-based cryptography, the Stern protocol \cite{Stern94} allowed to prove knowledge of a syndrome decoding instance, with a soundness error of $2/3$. As for the MPC-in-the-Head paradigm introduced in 2007 \cite{IKO}, the soundness error of the proof is of $1/N$ where $N$ is a chosen parameter, implying that the cheating probability can be made arbitrarily small. In the MPCitH protocol, a prover builds $N$ shares of his secret and simulates the execution of the MPC for every set of shares. Committing the executions and partially revealing some of them, the prover allows a verifier to check whether the computations were done honestly. Repeating this protocol multiple times guarantees the honesty of the prover. In recent years, the MPC-in-the-Head paradigm has been studied quite significantly. Several works aimed (and succeeded) at improving the MPCitH paradigm \cite{KKW18,BN20,FJR22,AGHH} (this list is far from exhaustive). Many schemes already exist based on the MPCitH, for instance, MQ \cite{mqom}, MinRank \cite{Mirath,BFG+}, Syndrome Decoding (Hamming or Rank) \cite{SDitHNIST,RYDE,BFG+}, and even more. In this work, we will use the Threshold-Computation-in-the-Head from \cite{FR23}, which leads to shorter proofs using Shamir's Secret Sharing, instead of the traditionally used additive sharing.

\textbf{The VOLEitH paradigm.} The VOLE-in-the-Head framework, introduced in \cite{BBd+} is also a framework to build proofs of knowledge. Essentially, it uses Vector Oblivious Linear Evaluations (VOLE) correlations to commit to a witness, and then prove the knowledge of the solution of polynomial constraints by performing operations on the hidden witness. While protocols already existed such as Quicksilver \cite{quicksilver} for instance, they did not allow to build signatures as they were in a designated verifier setting. The VOLE-in-the-Head framework improved this, by removing the designated verifier setting and making it publicly verifiable. The signature scheme FAEST \cite{FAEST}, submitted to the additional call for signatures from the NIST and selected for the second round, relies on this framework and is an efficient signature scheme relying on the security of AES.

\textbf{Our contribution.} We introduce in this work the Matrix Subcode Equivalence \new{Problem} and the Inhomogeneous Matrix Subcode Problem, which are related to the Matrix Code Equivalence Problem ($\MCE$) used in MEDS. These two problems can be seen as the rank metric version of the Hamming Subcode Equivalence problem and of the \new{Inhomogeneous} Permuted Kernel Problem. From the Inhomogeneous Matrix Subcode Problem, we build a signature scheme by using the recent MPC-in-the-Head techniques, to obtain a signature size smaller than MEDS for the first NIST security level, with much smaller public keys. We also detail the attacks on the problem, with a new formulation of the trilinear form algebraic attack from \cite{MEDS2} as well as a new algebraic modeling. The attacks are mostly the same as in $\MCE$, but their behaviour is widely different and require careful adaptation. In particular, we find that the attacks perform worse by a large margin, a difference closely related to the use of invariants. When looking at other signature schemes, ours compares quite well with other equivalence problems, as the public key is much smaller, but its signature size is larger than other MPC-in-the-Head schemes in rank metric, such as MinRank or Rank Syndrome Decoding (RSD). Furthermore, we compare well with SPHINCS+, with a smaller signature and signature + public key, which was the main comparison for the NIST call for additional signature schemes. We sum up the comparisons with other schemes in Table \ref{tab:comp}. 

\begin{table}[h]
\renewcommand*{\arraystretch}{1.5}
\centering
\scalebox{1}{
\begin{tabular}{|c|c|c|c|c|c|c|}
\hline
Scheme &~$|\sigma|$~ & ~$|\pk|$~ & ~$|\sigma| + |\pk|$~ \\
\hline
CROSS \cite{CROSS} & 8 960 B & 54 B & 9 014 B \\
MEDS \cite{medspqforum} & 5 200 B & 21 595 B & 26 795 B \\
ALTEQ \cite{CNR+} & 4 432 B & 357 256 B & 361 688 B \\
LESS \cite{CPS24} & 2 481 B & 13 939 B & 15 430 B\\
SPHINCS+ (short version) \cite{Sphincs} & 7 856 B  & 32 B & 7 888 B \\
RYDE \cite{RYDE} & 3 115 B & 69 B & 3 184 B \\
Mirath \cite{Mirath} & 2 990 B & 57 B & 3 047 B \\
\new{MQOM \cite{mqom}} & \new{2820 B} & \new{52 B} & \new{2 872 B} \\
Subfield Bilinear Collisions \cite{HJ24} & 2 786 B & 48 B & $\approx$ 2 834 B \\
Inhomogeneous Matrix Subcode Problem (this work)~ & \textbf{4 788 B} & \textbf{255 B} & \textbf{5 043 B} \\
\hline
\end{tabular}} \vspace{1.5mm}
\caption{Comparison with other signature schemes for the NIST security level I}
\label{tab:comp}
\end{table}

\textbf{Organization of the paper.} We present in section \ref{sec:preliminaires} notations and conventions, as well as some mathematical background. In section \ref{sec:MSE}, we introduce the Matrix Subcode Equivalence problem, prove its hardness, and give the MPC protocol in Section \ref{sec:protocol}. Finally, we describe and analyze the attacks on the problem in Section \ref{sec:attaques}, and give the choice of parameter, resulting sizes, and an idea of the computational cost in Section \ref{sec:param}. \new{Some details dealing with the signature scheme and the proof of knowledge can be found in Appendix \ref{app:tcith}.}

\section{Preliminaries and notation} \label{sec:preliminaires}
\subsection{Notation}

We denote by $\Fq$ the finite field of size $q$. The set of vectors with $n$ coordinates in $\Fq$ is referred as $\Fq^{n}$, the set of matrices with $m$ rows and $n$ columns in $\Fq $ is referred as $\Fq^{m\times n}$. We will also always consider $m\ge n$. We denote by $\GL_n(q)$ the set of invertible matrices of size $n \times n$ in $\Fq$, and $\mathsf{PGL}_n(q)$ the projective linear group. We use lowercase bold letters to represent vectors and uppercase bold letters for matrices ($\bm{x} \in \Fq^k$, $\bm{X} \in \Fq^{m \times n}$,  $x \in \Fq$). The subset of integers from $1$ to $n$ is represented with $\oneto{n}$. If $S$ is a set, we write $x\sampler S$ the uniform sampling of a random element $x$ in $S$. We denote by $\langle x_1,\dots,x_n\rangle$ the $\Fq$-linear subspace of $\Fqm$ generated by $(x_1,\dots,x_n) \in \Fqm^{n}$, and by $\mathcal{C}^{k}_{m \times n}$ the space of matrix codes of dimension $k$ with matrices of size $m \times n$. Let us define the gaussian coefficient $\CG{m}{r} = \prod_{i=0}^{r-1}\frac{q^m-q^i}{q^r-q^i}$ $ \underset{q\to + \infty}{\longrightarrow} q^{r(m-r)}$, it equals the number of distinct $r$-dimensional $\Fq$-linear subspaces of $\Fqm$.The set of all $r$-dimensional subspaces of a vector space $E$ will be designated by $\mathsf{Gr}^r(E)$. We denote by $\mathcal{S}_n$ the set of permutations of $\oneto{n}$. 
Two distributions $\{D_\lambda\}_\lambda$ and $\{E_\lambda\}_\lambda$ indexed by a security parameter $\lambda$ are $(t,\varepsilon)$-\textit{indistinguishable} (where $t$ and $\varepsilon$ are $\mathbb{N} \to \mathbb{R}$ functions) if, for any algorithm $\mathcal{A}$ running in time at most $t(\lambda)$ we have 
$$\big| \Pr[\mathcal{A}^{D_\lambda}\old{()}=1] - \Pr[\mathcal{A}^{E_\lambda}\old{()}=1 ]\big| \leq \varepsilon(\lambda)~, $$ 
with $\mathcal{A}^{Dist}$ meaning that $\mathcal{A}$ has access to a sampling oracle of distribution $Dist$. Finally, we use the notation $\share{x}_i$ to refer to the $i$-th share of a secret $x$ when using a secret sharing scheme.

\subsection{Secret Sharing}

We now define the secret sharing scheme that will be used in the protocols: Shamir's Secret Sharing scheme.
\begin{definition}[Shamir's Secret Sharing]\label{def:shamir-sharing}
Let $\mathbb{F}$ a field and $s \in \mathbb{F}$ a secret.
A Shamir's secret sharing is the following $(\ell+1,N)$-threshold sharing scheme: \begin{itemize}
\item Sample $(r_1, \ldots, r_\ell) \sampler \mathbb{F}^\ell$;
\item Compute $P(X) = s + \sum_{i=1}^\ell r_iX^i$;
\item Compute $\share{s}_i = P(e_i)$ where the $(e_i)_{i\in \oneto{N}}$ are distinct and non-zero public values.
\end{itemize}
\end{definition}

We define the \textit{degree} of a Shamir's secret sharing as the degree of the underlying polynomial. A sharing generated using the above process is of degree $\ell$. The sum of a $d_1$-degree sharing and a $d_2$-degree sharing is of degree $\le \max(d_1,d_2)$, while the multiplication is of degree $d_1+d_2$ 

\subsection{Matrix Codes}
It is necessary to recall some notions about matrix codes and matrix code equivalence. These notions can also be found in \cite{CDG21,MEDS1,MEDS2}, and other articles dealing with the subject.

An $[m\times n,k] $ matrix code $\C$ is a subspace of $\Fq^{m \times n}$ of dimension $k$, with basis $\langle \bm{C}_1,\dots,\bm{C}_k \rangle$ where the $\bm{C}_i$'s are all linearly independent elements of $\Fq^{m\times n}$. In this context, it is natural to use the rank metric for the weight, i.e, $w_R(\bm{M}) = \Rank(\bm{M})$, and $\operatorname{d} (\bm{A}\old{-}\new{,}\bm{B}) = \Rank(\bm{A}-\bm{B})$.

An \textit{isometry} $\mu$ is a linear map from $\Fq^{m\times n}$ to itself, which conserves the rank of an element, i.e, $\Rank(\mu({\bm{M}})) = \Rank(\bm{M})$, for all $ \bm{M} \in \Fq^{m\times n}$.

We say that two codes $\C$, $\D$ $\subset \Fq^{m\times n}$ are equivalent if there is an isometry $\mu$ such that $\mu(\C) = \D$. 
Since we are working with the rank metric, isometries are of the form $\bm{M} \mapsto \bm{A}\bm{M}\bm{B}$ where $\bm{A}$ and $\bm{B}$ are matrices of full rank, or $\bm{M} \mapsto \bm{M}^T$ (when $m=n$), or a composition of these two forms. In our case, only the first one is of interest and we will restrict to this case only.

\begin{definition} \label{def1}
    Let $\C$ and $\D$ $\subset \Fq^{m\times n}$ two matrix codes of dimension $k$. They are equivalent if and only if there exists $\bm{A} \in \mathsf{GL}_m(q)$ and $\bm{B} \in \mathsf{GL}_n(q)$ such that $\D=\bm{A}\C\bm{B}$ meaning for every $\bm{D} \in \D$, there is $\bm{C} \in \C$ such that $\bm{D} = \bm{A}\bm{C}\bm{B}$. 
\end{definition}

From this definition, one can easily define the Matrix Code Equivalence problem ($\MCE$) \new{\cite{Berger03}}:

\begin{problem} \label{MCE}$\MCE(q,m,n,k)$:\\
    Let $\C$ and $\D$ $\subset \Fq^{m\times n}$ two matrix codes. Decide if there exist matrices $\bm{A} \in \mathsf{GL}_m(q)$ and $\bm{B} \in \mathsf{GL}_n(q)$ such that $\D = \bm{A}\C\bm{B}$? 
\end{problem}

It is helpful to approach this problem by using vectors instead of matrices. For that purpose, we will give a bijection between $\Fq^{m\times n}$ and $\Fq^{mn}$:
\[ Vec_{row} \text{ }: \text{ }\bm{A} = \begin{pmatrix}
    a_{1,1} & \dots & a_{1,n} \\
    \vdots & \ddots & \vdots \\
    a_{m,1} & \dots & a_{m,n}
\end{pmatrix} \mapsto \vecrow{\bm{A}} = (a_{1,1} \dots a_{1,n} \dots a_{m,1}\dots a_{m,n})\]

Thanks to this bijection, we can view a [$m\times n,k$] matrix code $\C$ as a subspace of $\Fq^{mn}$ of dimension $k$ generated by a matrix $\bm{G} \in \Fqkmn$. We can also compute the dual of the code $\C$ generated by $\bm{G}$, by considering the dual matrix $\bm{H} \in \Fq^{(mn-k) \times mn}$ such that $\bm{G}\bm{H}^\top = \bm{0}$. Finally, this allows us to have the following lemma \cite{MEDS1}:
\begin{lemma} \label{lemma1}
    Let $\C$ and $\D$ be two {\rm[$m\times n,k$]} equivalent matrix codes over $\Fq$ such that $\D = \bm{A}\C \bm{B}$. Let $\bm{G}$ and $\bm{G}'$ be the two generator matrices of $\C$ and $\D$ respectively. Then, there exists a matrix $\bm{T} \in \GL_k(q)$ such that $\bm{T}^\top\bm{G}(\bm{A}^T\otimes \bm{B}) = \bm{G}'$ 
\end{lemma}

The matrix code equivalence problem has been studied in several articles \cite{CDG21,RST22,MEDS1}, and \cite{MEDS2}. This problem led to the MEDS signature scheme, which was submitted to the NIST call for digital signature schemes in 2023 \cite{MEDS2}.
\section{The Matrix Subcode Equivalence Problem} \label{sec:MSE}
Thanks to the previous notions, it is possible to define a new problem: similarly to the Subcode Equivalence problem in the Hamming metric, we define here the Matrix Subcode Equivalence problem ($\MSE$). We also define here the matrix code version of the Permuted Kernel Problem ($\IMSEP$).

\begin{problem} \label{problem2} $\MSE(q,m,n,k,k')$:\\
    Let $\C$ and $\D$ $\subset \Fq^{m\times n}$ be two matrix codes of dimension $k$ and $k'$, with $k' < k$ respectively. Decide if there exist matrices $\bm{A} \in \mathsf{GL}_m(q)$ and $\bm{B} \in \mathsf{GL}_n(q)$ such that $ \D \subset \bm{A}\C\bm{B}$. 
\end{problem}

We can thus adapt Lemma \ref{lemma1}:
\begin{lemma} \label{lemma2}
    Let $\C$ and $\D$ two matrix codes over $\Fq$ such that $ \D \subset \bm{A}\C\bm{B}$. Let $\bm{G}$ and $\bm{G}'$ be the two generator matrices of $\C$ and $\D$ respectively. Then, there exists a matrix $\bm{T} \in \Fq^{k\times k'}$ of rank $k'$ such that $\bm{G'} = \bm{T}^\top\bm{G}(\bm{A}^\top\otimes \bm{B})$.
\end{lemma}

Since $\bm{A}$ and $\bm{B}$ are invertible, one can easily see that this is equivalent to $\bm{G'}(\bm{A}^\top\otimes \bm{B})^{-1} = \bm{T}^\top\bm{G}$ which itself can be rewritten as $\bm{G'}(\bm{A}^\top\otimes \bm{B})^{-1}\bm{H}^\top = \bm{0}$ (where $\bm{H}$ is the parity-check matrix of $\bm{G}$). Note that, unlike the code equivalence problem, the matrix $\bm{T}$ in the subcode equivalence problem is not invertible (but has full rank).

Furthermore, knowing that $\MCE$ is at least as hard as the monomial equivalence problem \cite{CDG21}, the question can be asked on the hardness of the matrix subcode equivalence one. We then recall the Subcode Equivalence Problem ($\mathsf{SEP}$) in Hamming metric, which has been shown to be NP-complete in \cite{BTK17}:

\begin{problem}\label{SEP} $\mathsf{SEP}(q,n,k,k')$\\
    Let $\C \subset \Fq^n$ of dimension $k$, and $\D \subset \Fq^n$ of dimension $k'$ with $k' < k$ two codes in Hamming metric. Decide if there exists an isometry $\mu$ such that $\D \subset \mu(\C)$. In other words, given $\bm{G} \in \Fq^{k \times n}$ and $\bm{G}' \in \Fq^{k' \times n}$, decide if there exist $\mu$ and $\bm{T} \in \Fq^{k \times k'}$ of rank $k'$ such that $\bm{G}' = \bm{T}^\top \mu ( \bm{G})$.
\end{problem}

The MacWilliams equivalence theorem \cite{Mac62} says that the isometries in the Hamming metric are exactly the monomial transforms. Furthermore, the $\mathsf{SEP}$ is well-known to be NP-complete thanks to \cite{BTK17} when isometries are restricted to permutation. When the isometries are all the monomial transforms the problem is harder. We thus have the following result, which proves the NP-completeness of the $\MSE$:

\begin{proposition}
    The $\MSE$ is at least as hard as the Hamming $\mathsf{SEP}$.
\end{proposition}
\begin{proof}
We prove here that a solver for the $\MSE$ can also solve a Hamming $\mathsf{SEP}$ instance, by proceeding as in~\cite[Section 3.3]{CDG21}.

We first define $\psi$ as a map from $\Fq^n$ to $\Fq^{n \times n}$:
\[ \psi \text{ }: \text{ }\bm{x} = (x_1,\dots,x_n) \mapsto \psi(\bm{x}) = \begin{pmatrix}
    x_1 & & 0\\
    & \ddots &\\
    0& & x_n
\end{pmatrix} = \operatorname{Diag}(x_1,\dots,x_n).\]

This map translates the Hamming weight to the rank weight since we have $\mathsf{rank}(\psi(\bm{x})) = w_H(\bm{x})$. It is in fact an isometry on its image. Furthermore, if we define $\bm{y} = \pi(\bm{xD})$ for a vector $\bm{x} \in \Fq^n$, a diagonal matrix $\bm{D}\in\GL_n(q)$ and a permutation $\pi$ in $\mathcal{S}_n$, we have that $\psi(\bm{y}) = \bm{P}^{-1}\psi(\bm{x})\bm{DP}$, where $\bm{P}$ is the matrix in $\GL_n(q)$ representing the permutation $\pi$. 

Let $\C_H$ and $\D_H$ be two Hamming codes  of length $n$ and dimension $k$ and $k'$ respectively, and define the matrix codes $\C_M = \psi(\C_H)$ and $\D_M = \psi(\D_H)$, both subspaces of $\Fq^{n \times n}$, of dimension $k$ and $k'$ respectively. We will prove that $\C_H$ and $\D_H$ are \old{$\mathsf{SEP}$}\new{subcode}-equivalent if and only if $\C_M$ and $\D_M$ are \old{$\MSE$}\new{matrix subcode}-equivalent.

Indeed, if there exists a permutation $\pi$ and a diagonal matrix $\bm{D}$ such that $\D_H \subseteq \pi(\C_H\bm{D})$, then we have \[\D_M=\psi(\D_H)\subseteq\psi(\pi(\C_H\bm{D}))=\bm{P}^{-1}\psi(\C_H)\bm{DP}=\bm{P}^{-1}\C_M\bm{DP},\]
i.e. $\D_M$ is equivalent to a subcode of $\C_M$, with the isometry given by the matrices $(\bm{P}^{-1},\bm{DP})$. 

Conversely, assume $\D_M \subseteq \mu(\C_M)$ with $\mu$ an isometry. Note $\C'_M$ the $k'$ dimensional subspace of $\C_M$ s.t $\C'_M = \mu^{-1}(\D_M)$. As $\C_M$ is generated by diagonal matrices, it contains only diagonal matrices, hence is included in the image of $\psi$. As $\psi$ is injective, we can consider $\C'_H:=\psi^{-1}(\C'_M)$. This subcode is isometric to $\D_H$ since $\D_H=\psi^{-1}\circ\mu\circ\psi(\C'_H)$ and the applications $\mu$ and $\psi$ are isometries. 
\end{proof}

\old{It is well-known that the subcode equivalence problem reduces to a homogeneous PKP instance in the Hamming metric \cite{SBC24}. }\new{Using the primal formulation of the $\MSE$, we can construct the inhomogeneous version of the problem that we will call the Inhomogeneous Matrix Subcode Problem ($\IMSEP$). For that, let $\bm{G'} = \bm{TG}(\bm{A}^\top \otimes \bm{B})$. Then, one can easily imagine a modified instance where $\bm{G'} = (\bm{TG+G''})(\bm{A}^\top \otimes \bm{B})$ where $\bm{G''}$ is kept secret. Having $\bm{G''} = \bm{0}$ will thus correspond to $\MSE$, while $\bm{G''} \ne \bm{0}$ will correspond to the $\IMSEP$.}

\old{We now define the inhomogeneous version of the Matrix Code PKP in Problem \ref{problem3}.}
\begin{problem} \label{problem3} $\IMSEP(q,m,n,k,k')$:\\
    Let $\bm{G} \in \Fq^{k \times mn}$, $\bm{G}' \in \Fq^{k' \times mn}$ and $\bm{G}'' \in \Fq^{k' \times mn}$ be three generator matrices of two matrix codes of dimension $k$ and $k'$ respectively. Decide if there exists matrices $\bm{A} \in \mathsf{GL}_m(q)$, $\bm{B} \in \mathsf{GL}_n(q)$ and $\bm{T} \in \Fq^{k \times k'}$ of rank $k'$ and  such that $(\bm{T}^\top\bm{G} +\bm{G}'')(\bm{A}^\top\otimes \bm{B}) = \bm{G}'$.
\end{problem}

\begin{remark}
The previous problem can be reformulated as follows with the parity-check matrix $\bm{H} \in \Fq^{(mn-k) \times mn}$ of the code generated by $\bm{G}$ : let $\bm{H} \in \Fq^{(mn-k) \times mn}$, $\bm{Y} \in \Fq^{k' \times (mn-k)}$, and $\bm{G}' \in \Fq^{k' \times mn}$ be such that $\mathsf{rank}(\bm{H}) = mn-k$ and $\mathsf{rank}(\bm{G}') = k'$. Decide if there exist $\bm{A} \in \mathsf{GL}_m(q)$, $\bm{B} \in \mathsf{GL}_n(q)$ such that $\bm{G}'(\bm{A}^{-\top}\otimes \bm{B}^{-1})\bm{H}^\top = \bm{Y}$.
\end{remark}

It is easy to see that, when $\bm{Y} = \bm{0}$, $\IMSEP$ is in fact only an $\MSE$ instance. In fact, if $\bm{Y} = \bm{0}$, then this implies that $\bm{G}'(\bm{A}^{-\top}\otimes \bm{B}^{-1})\bm{H}^\top = \bm{0}$, which is the $\MSE$. However, having a non-zero $\bm{Y}$ will be useful in the signature protocol. Furthermore, there is a way to solve a $\IMSEP$ instance given an $\MSE$ solver.

\begin{lemma}\label{IMSEPtoMSE}
Let $q,m,n,k,k'$ be positive integers with $k' < k$. Let $\mathcal{A}$ be an algorithm which solves a $(q,m,n,k+k',k')$ $\MSE$ instance with probability $\epsilon_{1}$. Then, there is an algorithm $\mathcal{A}'$ that solves a $(q,m,n,k,k')$ $\IMSEP$ instance with probability $\epsilon_{2}$ where $\epsilon_2 \ge \epsilon_1$
\end{lemma}

\begin{proof}\hfill \vspace{-0.7cm}
\begin{algorithm}[H]
\caption{Algorithm $\mathcal{A}'$ to solve an $\IMSEP$ instance}
\label{Reduction}
\begin{algorithmic}[1]
\REQUIRE $\bm{G}'$, $\bm{H}$, $\bm{Y}$ an $\IMSEP$ instance. 
\ENSURE $1$ if the instance has a solution, $0$ else.
\STATE Find $\bm{G}'' \in \Fq^{k' \times mn}$ such that $\bm{G}'' \bm{H}^\top = \bm{Y}$.
\STATE Define $\bm{G}_{aug} := \begin{bmatrix}\bm{G}\\\bm{G}'' \end{bmatrix}$ and compute its dual matrix $\bm{H}_{aug} \in \Fq^{mn-k-k' \times mn}$.
\STATE Run $\mathcal{A}$ on input $(\bm{G}',\bm{G}_{aug})$ and obtain a value $b$.
\STATE return $b$.
\end{algorithmic} 
\end{algorithm}
Let $\bm{G}', \bm{H}^\top$, and $\bm{Y} \ne \bm{0}$ be a $(q,m,n,k,k')$ $\IMSEP$ instance using the dual formulation. Let $\bm{G}'' \in \Fq^{k' \times mn}$ be such that $\bm{G}'' \bm{H}^\top = \bm{Y}$. Then, $\bm{G}'' = \bm{G}'(\bm{A}^{-\top}\otimes \bm{B}^{-1}) + \bm{G}_e$ for a matrix $\bm{G}_e = \bm{T}^\top \bm{G}$ with $\bm{T} \in \Fq^{k \times k'}$ (because $\bm{T}^\top \bm{G} \bm{H}^\top= \bm{0}$). Let $\bm{G}_{aug} := \begin{bmatrix}\bm{G}\\\bm{G}'' \end{bmatrix}$ be the augmented matrix. Then, the code $\bm{G}'$ and the code $\bm{G}_{aug}$ form an instance of the $\MSE$ with parameters $(q,m,n,k+k',k')$, using the very same isometry ($\bm{A},\bm{B}$), meaning Algorithm $\mathcal{A}'$ is correct.
\end{proof}

Finally, one crucial aspect of the problem is the number of solutions.
We will consider the action of the isometries 

\begin{align*}
    \star \text{ } : \text{} (\GL_m(q) \times \GL_n(q)) \times \mathcal{C}^{k}_{m \times n} &\mapsto \mathcal{C}^{k}_{m \times n}\\
((\bm{A},\bm{B}),\C) &\mapsto \bm{A}\C\bm{B}
\end{align*}

While in the $\MCE$ case, it suffices to consider that the automorphism group of the above action group is trivial (up to scalar multiplication) to have a unique solution, for $\MSE$ or $\IMSEP$, some precisions are needed (we justify here with $\MSE$ but the same apply with $\IMSEP$). 

We will here give an argument quite similar to what was done for PKP in \cite{SBC24}.
Concretely, in $\MSE$, we are looking for matrices $\bm{A} \in \GL_m(q)$ and $\bm{B} \in \GL_n(q)$ such that $\bm{G}'\left( \bm{A}^\top\otimes \bm{B}\right)^{-1}$ generates a subspace of $\C$ of dimension $k'$. There are thus $\vert 
\mathsf{PGL}_m(q)\vert \cdot\vert \mathsf{PGL}_n(q)\vert$ possible subspaces at most. There are also $\CG{k}{k'}$ subspaces of dimension $k'$ in the code generated by $\bm{G}$, and $\CG{mn}{k'}$ total subspaces of dimension $k'$. Then, assuming 
$\bm{G}'\left(\bm{A}^\top\otimes \bm{B}\right)^{-1}$ behaves as a random subspace of 
dimension $k'$ we get that the number of solutions on average is $$\vert 
\mathsf{PGL}_m(q)\vert \cdot\vert \mathsf{PGL}_n(q)\vert \cdot \frac{\CG{k}{k'}}{\CG{mn}{k'}}$$

As a result, we will take parameters where the above value is less than or equal to $1$, to get on average only one solution.

We now explain the minimal subcode size that must be chosen for an $\MSE$ instance (and thus the minimal value of $k'$ for $\IMSEP$). \new{Assume $m = n$. Let $\D$ be a subcode of dimension $1$, i.e., $\D = \langle \bm{D}_1 \rangle$, and $\C$ a code of dimension $k$, i.e., $\C = \langle \bm{C}_1,\dots,\bm{C}_k \rangle$. Then, by taking as unknowns the matrices $\Tilde{\bm{A}}$ and $\Tilde{\bm{B}}$, and any matrix $\bm{C} \in \C$ of rank $\Rank(\bm{D_1})$, the system $\Tilde{\bm{A}}\bm{D}_1 = \bm{C}\Tilde{\bm{B}}$ is clearly underdetermined as there are at most $n^2$ independent linear equations for $2n^2$ unknowns.} By using simple Gaussian elimination, a basis of the solutions can be found easily. Given the large solution space, it is very likely that multiple solutions exist in which both matrices $\bm{A}$ and $\bm{B}$ are invertible, thereby increasing the expected number of solutions of the $\MSE$ instance. This means that the subcode must be taken with a dimension higher than $1$. The same reasoning can be applied with a subcode of dimension $2$: let  $\D = \langle \bm{D}_1, \bm{D}_2 \rangle$, and let $ \bm{C}_1$, $\bm{C}_2$ be two matrices of the bigger code $\C$ of respective rank $\Rank(\bm{D_1}),\Rank(\bm{D_2})$. Then, the system 
\begin{equation*}
\begin{cases}\Tilde{\bm{A}}\bm{D}_1 = \bm{C}_1\Tilde{\bm{B}}\\\Tilde{\bm{A}}\bm{D}_2 = \bm{C}_2\Tilde{\bm{B}}\end{cases}
\end{equation*} 
can be underdetermined as there are \new{\textit{at most} $2n^2$ independent linear equations for $2n^2$} unknowns, which results in the same process where several solutions can exist. This phenomenon does not appear when \new{$3 \leq k' <n$, as such systems are then overdetermined. However one could try to find specific matrices $\bm{C}_1,\ldots,\bm{C}_{k'}$ such that enough redundancy appear among the system of equations. This would implies that the matrices $\bm{D}_i$ and $\bm{C}_i$ for $1\leq i\leq k'$ have a lower rank $r$ than $n$. This is a MinRank Problem for which we have two constraints. First the MinRank Problem needs to be \emph{under-defined} which occurs whenever $k'\geq (n-r)^2$ and secondly the rank $r$ must verify $rk'n \leq 2n^2$ (less equations than variables). The first constraint $k'\geq (n-r)^2$ implies $r$ at least equal to $r_0=n-\sqrt{k'}$. Assuming $k'<n$, $r_0$ is lower bounded by $n(1-\frac{1}{\sqrt{n}})$. When $k' \geq 3$ and $n\geq 9$, the minimum number of equations would be $k'\ceil{r_0}n > k'n^2(1-\frac{1}{\sqrt{n}}) \geq 2n^2$. Thus we can't find matrices in both codes such that the previous system gives less equations than variables. For our parameters, we will respect these constraints $3 \leq k' <n$, $m=n$ and $n\geq 9$.}

We give in the following section the Multi-Party-Computation (MPC) protocol for the signature, and we detail the attacks on the $\MSE$ in section \ref{sec:attaques}.

\section{MPC modelings and application to TCitH and VOLEitH frameworks} \label{sec:protocol}

We describe in this section the MPC-in-the-Head paradigm \old{and the variants of this paradigm we will consider}\new{and its most recent improvements}. We only give a brief high-level overview of the techniques, as the frameworks are the same as all the previous work on the subject. However, we will detail the polynomial constraint checking protocol that will be used in our signature scheme.

The MPC-in-the-Head paradigm, introduced in \cite{IKO}, is an extremely powerful tool to build proofs of knowledge. This paradigm relies on secure multi-party computations and is especially efficient for proving the knowledge of code-based problems. Essentially, an MPC-in-the-Head proof of knowledge works as follows:
\begin{itemize}
    \item The prover shares his witness (his secret) in $N$ shares using an $(\ell+1,N)$ threshold sharing scheme. He commits to these $N$ shares and sends the commitments to the verifier.
    \item The prover performs the MPC protocol "in his head", and sends a commitment of the computations;
    \item The verifier asks for $\ell$ shares of the witness, as $\ell+1$ shares are necessary to retrieve it;
    \item The prover sends the shares of the $\ell$ parties. From them, the verifier can check the consistency of $\ell$ of the commitments. \end{itemize}

The correctness is straightforward, while the zero-knowledge property is guaranteed because no more than $\ell$ shares of the witness are revealed. The soundness property is also guaranteed since a malicious prover will need to cheat on at least one share as he will not know the correct value of the secret. Finally, before performing the MPC protocol "in his head", the prover receives a challenge from the verifier that he will use in his computations so that he cannot cheat in the MPC protocol.

This paradigm has known many improvements and optimizations over the years (e.g. \cite{limbo,KKW18,BN20,FR22,FR23} for a non-exhaustive list). We recall hereafter two recent frameworks: the \textit{Threshold-Computation-in-the-Head} from \cite{FR23}, and the \textit{VOLE-in-the-Head} from \cite{BBd+}.

\subsection{The Threshold-Computation-in-the-Head and VOLE-in-the-Head frameworks}
\new{The TCitH framework has been introduced in \cite{FR23}, which is an improvement of \cite{FR22}. The idea of this framework is to rely on Shamir's Secret Sharing instead of the usually used additive sharing. The main difference introduced by using Shamir's Secret Sharing consists in a better handling of the multiplication of shares. In fact, by using Shamir's Secret Sharing on two secrets $a$, $b \in \Fq$, it is well established that $\share{a}\cdot \share{b} = \share{a\cdot b}$ while it is not the case with additive secret sharing, thus making the handling of multiplications much easier. In particular, \cite{FR23} introduced the polynomial constraints checking protocol (explained below), allowing one to obtain a signature scheme only from the description of the polynomial constraints associated to a problem. We will thus use this framework in a "black box" way, by describing only the polynomial constraints and the resulting sizes. We give a more thorough description of the framework in Appendix \ref{app:tcith} and refer to \cite{FR23} for more details.}

\new{The VOLEitH framework, introduced independently in \cite{BBd+} from the TCitH one, can also be expressed in the same formalism. Very roughly, it consists in merging some of the shares considered in the TCitH framework, allowing to make fewer computations, but in a much bigger field. We also postpone details in Appendix \ref{app:tcith} as we also use this framework as a "black box".}

\new{Finally, both frameworks use an MPC protocol for polynomial checking (we follow the description presented in \cite[Section 6]{FR23}), i.e, a protocol such that, given $m$ polynomial $f_1,\dots,f_m \in \Fq[x_1,\dots,x_n]$, one can check that a witness $\bm{x} \in \Fq^n$ is such that $f_1(\bm{x}) = \dots = f_m(\bm{x}) = 0$. In a nutshell, the protocol evaluates $f_j(\share{\bm{x}}_i)$ for all $j$ in $\oneto{m}$, and then computes a linear combination $\share{\alpha}_i= \sum_{j=1}^{m}\gamma_jf_j(\share{\bm{x}}_i)$ (note that it is easy to evaluate the polynomials $f_j$, because we are using Shamir's secret sharing). The obtained value $\alpha$ must be equal to $0$ if the witness is valid. (See Appendix \ref{app:tcith}, \cite{FR23,BBd+} for more details).}

\new{A straightforward example of this protocol is with the Multivariate Quadratic problem, which asks, given $m$ polynomials in $n$ variables $f_1,\dots,f_m$, to find a vector $\bm{x} \in \Fq^n$ such that $f_1(\bm{x}) = \dots = f_m(\bm{x}) = 0$. In the following section, we will describe the polynomials that must be equal to $0$ when considering the $\IMSEP$.}

\subsection{Polynomial constraints for the Inhomogeneous Matrix Subcode Problem}

In both these frameworks, a crucial part that influences the signature size is the MPC protocol used. In particular, the part of the MPC protocol that has the greatest influence is the modeling used, i.e., how the witness is represented. We describe below the constraints for $\IMSEP$ that must be checked in both of the frameworks.

The modeling for $\IMSEP$ is straightforward: the parties receive a share of $\bm{A}$ and $\bm{B}$, and they then perform the computation to verify the value of $\bm{Y}$, i.e, they want to verify that $\bm{G}'(\bm{A}^\top \otimes \bm{B})\bm{H}^\top = \bm{Y}$ (although Problem \ref{problem3} uses the notation with the inverses, we simply use $\bm{A}$ and $\bm{B}$ here, to lighten the notations). Doing so, the witness size does not depend on the dimension of the code and requires only shares of $\bm{A}$ and $\bm{B}$. This makes an efficient protocol with witness size $(n^2+m^2)\log_2(q)$. We describe the MPC protocol associated with this modeling in Figure \ref{fig:TCitHMPCIMSEP}. Finally, the invertibility of the matrices $\bm{A}$ and $\bm{B}$ does not need to be verified due to taking $\bm{Y}$ of rank $k'$ (which is the case when the generation of the instance is done honestly).

In fact, assuming $\bm{G}'\left(\bm{A}^T\otimes \bm{B}\right)$ behaves as a random subspace of dimension $\le k'$, there are, on average, $$\frac{q^{m^2+n^2}}{q-1} \cdot \frac{\CG{k}{k'}}{\CG{mn}{k'}} $$
couples $\bm{A} \in \Fq^{m \times m},\bm{B}\in \Fq^{n \times n}$ such that $\bm{G}'\left(\bm{A}^T\otimes \bm{B}\right)\bm{H}^\top = \bm{Y}$ with $\bm{Y}$ of rank $k'$. We will always take parameters such that this quantity is negligible. Note that we use $\frac{q^{m^2+n^2}}{q-1}$ because if $(\bm{A},\bm{B})$ is a solution, so is $(\lambda \bm{A}, \lambda^{-1}\bm{B})$ for all $\lambda \in \Fq^\times$. Thus, no solution exists where $\bm{A}$ and $\bm{B}$ are not invertible, which avoids having to verify this property. 
\begin{algorithm}[h]
\caption{MPC protocol of degree $2$, verifying that a given input is a solution of an $\IMSEP$ instance for the $i$th party}
\label{fig:TCitHMPCIMSEP}
\begin{algorithmic}[1]
\PUBLICVAL An instance of the $\IMSEP$, $\bm{H} \in \Fq^{mn-k\times mn}$, $\bm{G}' \in \Fq^{k' \times mn}$, $\bm{Y} \in \Fq^{k' \times mn-k}$. 
\REQUIRE A pair of shares of $\bm{A} \in \Fq^{m \times m}$ and $\bm{B}\in \Fq^{n \times n}$: $\share{\bm{A}}_i,\share{\bm{B}}_i$.
\ENSURE If $\bm{A}$ and $\bm{B}$ are solutions of the $\IMSEP$ instance, outputs $\accept$.

\STATE The party receives sharings $\left(\share{\bm{A}}_i,\share{\bm{B}}_i \right)$  where $\deg(\share{\bm{A}}_i) = \deg(\share{\bm{B}}_i) = 1$.
\STATE The party receives a random sharing $\share{v}_i$ of $v=0$ in $\Fq$ and $\deg(\share{v}_i) = 2$. 
\STATE The party receives random values $\gamma_1, \ldots, \gamma_{k'(mn-k)} \in \mathbb{F}_{q}$. 
\STATE The party computes $\share{\bm{\Tilde{G}}'}_i = \bm{G}'\left(\share{\bm{A}}_i^\top \otimes \share{\bm{B}}_i\right)\bm{H}^\top - \bm{Y}$.
\For{\text{$j \in \oneto{k'(mn-k)}$}}
\STATE The party computes $\share{f_j}_i = \share{\vecrow{\bm{\Tilde{G}}'}_{j}}_i $.
\EndFor
\STATE The party computes $\share{\alpha}_i = \share{v}_i + \sum\nolimits_{j = 1}^{k'(mn-k)} \gamma_j \cdot \share{f_j}_i$.
\STATE The party opens $\share{\alpha}_i$ to get $\alpha$.
\STATE  The party outputs $\accept$ if and only if $\alpha = 0$. Else, output $\reject$
\end{algorithmic} 
\end{algorithm}

\begin{lemma}
The false-positive probability of the protocol from Figure \ref{fig:TCitHMPCIMSEP} is $\frac{1}{q}$ for one repetition, and $\frac{1}{q^\rho}$ when repeated $\rho$ times.
\end{lemma}
\begin{proof}
The proof can be found in \cite[Lemma 2]{FR23}. Essentially, if there is at least one $f_i$ such that $f_i(\bm{x}) \ne 0$, then $\alpha$ is uniformly random in $\Fq$ as all the $\{\gamma_j\}_{j \in \oneto{k'(mn-k)}}$ are uniformly sampled.
\end{proof}

Due to being a protocol where the degree of the sharings is only $2$, $\tau$ is taken such that $\frac{2}{N^\tau} \le 2^{-\lambda+w}$ for VOLEitH, and such that $\left(\frac{2}{N}\right)^\tau \le 2^{-\lambda+w}$ for TCitH. 

{\bf \noindent Transformation into a signature.} Finally, we stress that the transformation of this protocol into a signature scheme is very generic and is not detailed in the purpose of conciseness. \new{Roughly speaking, the signature is then composed of all the 'symmetric' parts, i.e., elements that are inherent to the MPC-in-the-Head framework used, as well as some shares of the secret (shares of $\bm{A}$ and $\bm{B}$, of size $\log_2(q)(m^2+n^2)$ in total) and of $\alpha$. We refer to Appendix \ref{app:tcith} and \cite{FR23,BBd+} for more details.}

\begin{remark}
An essential point of using $\IMSEP$ instead of $\MSE$ in the protocol is the following. Let $\bm{G}'\left( \bm{A}^\top \otimes \bm{B} \right)\bm{H}^\top= \bm{0}$ be the equation to be verified in the protocol. Then, a prover could easily cheat by picking $\bm{A}=\bm{B}=\bm{0}$. Furthermore, specializing one variable in $\bm{A}$ and $\bm{B}$ would not suffice, as many ways to cheat could appear. For instance, suppose $\bm{G}'$ is of the form $\begin{bmatrix}\bm{I}_{k'} & \bm{G}'' \end{bmatrix}$, where the $j$th column of $\bm{G}''$ is the vector $(1,1,\dots,1)^\top$ (the approach also works with a multiple of this vector). Then, a prover can easily cheat, by taking $\bm{A}$ to be the null matrix except on its first top left coefficient, i.e, $a_{1,1} = 1$, and $\bm{B}$ such that $b_{i,1} = 1$ for $i \in \oneto{k'}$ and $b_{k'+j,1} = -1$. In fact, doing so, $\bm{G}'\left( \bm{A}^\top \otimes \bm{B} \right) = \bm{0}$, resulting in a trivial solution for such a subcode. Although only an example, this cheating method illustrates the need to have an "inhomogeneous" instance of the problem. Thus, using $\MSE$ would result in the need to verify the subcode in another way: verifying that $\bm{G}'\left( \bm{A}^\top \otimes \bm{B} \right)$ is equal to $\bm{T}^\top \bm{G}$ for some matrix $\bm{T} \in \Fq^{k \times k'}$ of rank $k'$. \old{The signature size would be greater, as we will see in Section \ref{sec:param}, due to having to transmit the matrix $\bm{T}$ in addition to $\bm{A}$ and $\bm{B}$.}
\end{remark}

\subsection{Key generation}

The key generation algorithm is straightforward: $\bm{H}$, $\bm{G'}$, $\bm{A}$, and $\bm{B}$ are sampled randomly, which allows to compute $\bm{Y}$. The procedure is described in Algorithm \ref{KeyGen}. 
\begin{algorithm}[h!]
\caption{$\mathsf{KeyGen}$ algorithm for a $\IMSEP$ instance}
\label{KeyGen}
\begin{algorithmic}[1]
\REQUIRE Parameters $(q,m,n,k,k')$ of an $\IMSEP$ instance. 
\ENSURE $\sk =  (\bm{A},\bm{B})$ and $\pk = (\bm{G},\bm{Y})$.
\STATE $\bm{H} \sampler \Fq^{mn-k \times mn}$ with $\mathsf{rank}(\bm{H}) = mn-k$.
\STATE $\bm{G}' \sampler \Fq^{k' \times mn}$
\STATE $\bm{A} \sampler \GL_m(q)$ and $\bm{B} \sampler \GL_n(q)$.
\STATE $\bm{Y} :=  \mathsf{RowReducedEchelonForm}\big(\bm{G}'\left(\bm{A}^\top \otimes \bm{B} \right)\bm{H}^\top\big)$.
\IF{$\mathsf{Rank(\bm{Y})} < k'$}
\STATE Go back to Step $1$
\ENDIF
\STATE return $\sk = (\bm{A},\bm{B})$ and $\pk = (\bm{H},\bm{G}',\bm{Y})$.
\end{algorithmic} 
\end{algorithm}

The size of $\sk$ is simply the seed used for the sampling. For $\pk$, it is of size $\lambda + \log_2(q) \cdot k' \cdot (mn-k)$ as the matrix $\bm{Y}$ can be stored in a row reduced echelon form. We already see that thanks to the use of $\IMSEP$, the public key will be much smaller than if the code equivalence problem was considered. 

\section{The attacks on Matrix Subcode Equivalence} \label{sec:attaques}
We now describe the attacks on the $\MSE$ as $\IMSEP$ is attacked by using the reduction to $\MSE$ (note that we disregard here the complexity of the reduction from $\IMSEP$ to $\MSE$). Since $\mathsf{MSEP}$ and $\mathsf{MCEP}$ are very similar, we will use the same attacks from code equivalence, and adapt them to the subcode case.

\subsection{Attack through code equivalence}
A first and naive attack is to simply try and guess the subcode taken in $\C$, and then solve it as an $\mathsf{MCE}$ instance. Since there are $\CG{k}{k'}$ possible subspaces, the complexity is
$$ \CG{k}{k'} \mathsf{C}_{\MCE_{(q,m,n,k')}} \quad $$ where $\mathsf{C}_{\MCE_{(q,m,n,k')}}$ is the complexity of solving an $\MCE$ instance with parameters $m,n,k'$.
This is not efficient however: while the cost of solving the $\MCE$ instance is quite low as it is a much easier one, the quantity $\CG{k}{k'}$ is very large, even for small parameters, as it can be approximated to $q^{k'(k-k')}$.  ~(For instance, with $q=64$, $k'=3$, and $k=12$, this value is already at $2^{162}$, while our actual parameters are even bigger). This makes the attack quickly not practical and unusable. 

\subsection{Algebraic attacks} \label{subsec:algatk}
\paragraph{Naive modeling.} A first algebraic modeling is to simply consider $\bm{G'} = \bm{T}^\top\bm{G}(\bm{A}^\top\otimes \bm{B})$ with $\bm{T}$, $\bm{B}$, and $\bm{A}$ as unknowns. We then have $k'k + m^2 + n^2$ unknowns, and $k'mn$ equations that are (affine) trilinear.
\paragraph{Hybrid approach.} Using the degree 2 modeling $\bm{G'}(\bm{I_m}\otimes \bm{B}^{-1})=\bm{T}^\top\bm{G}(\bm{A}^\top\otimes\bm{I_n})$, it is possible to do as in \cite{MEDS1}, and guess the $\alpha$ first columns of $\bm{A}^\top$. We will obtain $k'\alpha n$ linear equations in $n^2+kk'$ unknowns. This means that when $\alpha = \ceil{\frac{n^2+kk'}{k'n}}$, it is possible to solve the linear system uniquely. We obtain a complexity of $$\mathcal{O}\left( q^{m\ceil{\frac{n^2+kk'}{k'n}}}(n^2+kk')^{\omega} \right), \text{~~~~} n,k \rightarrow \infty$$ 
When $k'$ is small, we see that this approach is much less effective than in the case of $\MCE$.

It is also possible to do the same process with the rows of $\bm{T}^\top$: one can guess the $\alpha$ first rows of $\bm{T}^\top$, to obtain $\alpha mn$ linear equations, in $n^2+m^2$ variables.
The complexity is similar: $$\mathcal{O}\left( q^{k\ceil{\frac{n^2+m^2}{mn}}}(n^2+m^2)^{\omega} \right), \text{~~~~} n,k \rightarrow \infty$$
which will not be effective, since we will take $k$ quite large.

\paragraph{Dual modeling.} Using the same modeling as used in \cite{MEDS1}, it is possible to remove the variables from $\bm{T}$. Let 
$\Tilde{\bm{G'}} = \bm{G'}(\bm{A}^{-\top}\otimes \bm{B}^{-1})$
where the unknowns are the coefficients of $\bm{A}^{-1}$ and $\bm{B}^{-1}$. Then, one easily sees that, if we note $\bm{H}$ the dual matrix of $\bm{G}$ (i.e, $\bm{G}\bm{H}^\top = \bm{0}$), then $$\bm{G'}(\bm{A}^{-\top}\otimes \bm{B}^{-1}) \bm{H}^\top = \bm{0}$$ 
Since $\bm{H} \in \Fq^{(mn-k) \times mn}$, this system possesses $m^2+n^2$ unknowns, and $k'(mn-k)$ quadratic equations. \old{A priori, } This attack \old{should be} \new{is} less effective than for $\MCE$, as we have the same number of unknowns, but much fewer equations. Experiments confirm this intuition.

\paragraph{Trilinear form modeling.} Finally, we describe a last algebraic attack, which is inspired by \cite{MEDS2}. We use the notation $\vecrow{\bm{M}}$ for the vector formed by the concatenation of the rows of a matrix $\bm{M}$. Using the commutation matrices\footnote{The commutation matrix $\bm{K}_{p,q}$ is the permutation matrix of dimension $pq$ such that, for any matrices $\bm{A}\in\Fq^{p_1\times q_1}$ and $\bm{B}\in\Fq^{p_2\times q_2}$, we have
$    \bm{K}_{p_1,p_2}^\top(\bm{A}\otimes\bm{B})\bm{K}_{q_1,q_2} = \bm{B}\otimes\bm{A}$.
}  such that
\begin{align*}
  (\bm{T}\otimes\bm{A}^\top\otimes \bm{B})=\bm{K}_{mn,k}^\top(\bm{A}^\top\otimes \bm{B} \otimes \bm{T})\bm{K}_{mn,k'}=\bm{K}_{n,km}^\top(\bm{B}\otimes \bm{T}\otimes \bm{A}^\top)\bm{K}_{n,k'm},
\end{align*}
we define $\bm{G_{A}}\in\Fq^{m\times nk}$ such that 

$\vecrow{\bm{G_{A}}}=\vecrow{\bm{G}}\bm{K}_{mn,k}^\top$
and $\bm{G'_{A}}\in\Fq^{m\times nk'}$ such that 

$\vecrow{\bm{G'_A}} = \vecrow{\bm{G'}}\bm{K}_{mn,k'}^\top$. 
The matrices $\bm{G_{B}}\in\Fq^{n\times km}$ and $\bm{G'_{B}}\in\Fq^{n\times k'm}$ are defined analogously with the commutation matrices $\bm{K}_{n,km}^\top$ and $\bm{K}_{n,k'm}^\top$. Then, the naive modeling rewrites in three equivalent sets of equations: 
\begin{align}
    \vecrow{\bm{G'}}=\vecrow{\bm{G}}(\bm{T}\otimes\bm{A}^\top\otimes \bm{B}) & \Leftrightarrow  \bm{G'}=\bm{T}^\top \bm{G}(\bm{A}^\top\otimes \bm{B})\label{eq:TAB}\\
    \vecrow{\bm{G'_A}}=\vecrow{\bm{G_A}}(\bm{A}^\top\otimes \bm{B}\otimes\bm{T})& \Leftrightarrow  \bm{G'_A}=\bm{A} \bm{G_A}(\bm{B}\otimes \bm{T})\label{eq:ABT}\\
    \vecrow{\bm{G'_B}}=\vecrow{\bm{G_B}}(\bm{B}\otimes \bm{T}\otimes\bm{A}^\top)& \Leftrightarrow  \bm{G'_B}=\bm{B}^\top \bm{G_B}(\bm{T}\otimes \bm{A}^\top)\label{eq:BTA}
\end{align}
With these notation and using the dual matrix $\bm{H'_{A}}\in\Fq^{(nk'-m)\times nk'}$ (resp. $\bm{H'_{B}}\in\Fq^{(mk'-n)\times mk'}$) associated to $\bm{G'_{A}}$ (resp. $\bm{G'_{B}}$), we recover the $m(nk'-m)$ (resp. $n(mk'-n)$) linearly independent bilinear equations obtained in \cite{MEDS2} via trilinear forms:
\begin{align*}
    0 &= \bm{G_A}(\bm{B}\otimes\bm{T})\bm{{H'}_A^\top}& m(nk'-m) \text{ equations in } \bm{B},\bm{T}, \\
    0 &= \bm{G_B}(\bm{T}\otimes\bm{A}^\top)\bm{{H'}_B}^\top& n(mk'-n) \text{ equations in } \bm{A},\bm{T}.
\end{align*}

However, as $\bm{T}$ is no longer invertible, we do not get bilinear equations in $\bm{A},\bm{B}$ as in \cite{MEDS2}. One could consider the $k'(nm-k)$ equations  $$0 = \bm{G'_T}(\bm{A}^{-1}\otimes\bm{B}^{-\top})\bm{{H}_T}^\top$$ by adding $m^2+n^2$ variables $A^{-1},B^{-1}$ which is not worth. We won't use it in our modeling.

\paragraph{New algebraic modeling.} This coding point of view led us to a new set of trilinear equations, that are experimentally independent from the bilinear equations. For any matrix $\bm{M}$ with linearly independent rows, let $\rightinverse{\bm{M}}$ be a right inverse of $\bm{M}$, i.e. a matrix such that $\bm{M}\rightinverse{\bm{M}}=\bm{I}$. By expressing the fact that $\bm{A}$ (resp. $\bm{B} $) commutes with its inverse, we get the following equations:
\begin{align}
\eqref{eq:ABT} \Rightarrow    \bm{I_m}&=\bm{A} \underbrace{\bm{G_A}(\bm{B}\otimes \bm{T}) \rightinverse{\bm{G'_A}} }_{=\bm{A}^{-1}}& \Rightarrow
   \bm{I_m}&= \bm{G_A}(\bm{B}\otimes \bm{T}) \rightinverse{\bm{G'_A}} \bm{A}.\label{eq:m2}\\ 
 \eqref{eq:BTA} \Rightarrow  \bm{I_n}&=  \bm{B}^\top \bm{G_B}(\bm{T}\otimes \bm{A}^\top)\rightinverse{\bm{G'_B}} & \Rightarrow
\bm{I_n}&=  \bm{G_B}(\bm{T}\otimes \bm{A}^\top)\rightinverse{\bm{G'_B}} \bm{B}^\top. \label{eq:n2}
\end{align}
The equations on the left are linear combinations of \eqref{eq:ABT} or \eqref{eq:BTA}, but the equations on the right are linearly independent from them. In practice, we get $n^2+m^2-1$ new equations that contain equations~\eqref{eq:TAB}--\eqref{eq:BTA}. $\bm{T}$'s right inverses are no use since they don't commute with $\bm{T}$. If moreover $\bm{T}$ is invertible, we get $k^2-1$ additional equations. Here we manage to use the invertibility of $\bm{A}$ and $\bm{B}$ without adding the equations $\bm{A}\bm{A}^{-1} -\bm{I_m}$ or $\bm{B}\bm{B}^{-1} -\bm{I_n}$ and thus adding new variables for $\bm{A}^{-1},\bm{B}^{-1}$ or expressing the inverse thanks to the comatrices which would lead to equations of degree $m$ or $n$.

\paragraph{Experimental results.} Experimentally, the naive trilinear equations are linearly dependent from the new trilinear equations~\eqref{eq:m2} and~\eqref{eq:n2}, which are all but one linearly independent: we add $n^2+m^2-1$ linearly independent trilinear equations to the system (and $(k^2-1)$ additional ones if $k=k'$). Moreover, if $(\bm{A},\bm{B},\bm{T})$ is a solution to the system, then $(a\bm{A},b\bm{B},t\bm{T})$ is also a solution for any $a,b,t\in\Fq$ such that $abt=1$. This means that we can with high probability assume that $\bm{A}_{1,1}=1$ and $\bm{T}_{1,1}=1$ and get a unique solution to the system.

Choosing subcode instead of code reduces the number of equations available in the modeling. Indeed, the number of quadratic equations decreases from $n(mk-n) + m(nk-m) + k(nm-k)$ to $n(mk'-n) + m(nk'-m)$ and the number of trilinear equations from the new modeling decreases from $n^2+m^2+k^2-2$ to $n^2+m^2-1$ when the number of variables decreases from $n^2+m^2+k^2$ to $n^2+m^2+kk'$. 
For the small examples we were able to run with parameters $n=3,m=4,k=4,q=64$ the Gröbner basis computation took around 100 seconds for $k'=k$ and more than 10 hours for $k'=k-1$.

\paragraph{Constant rate.} \new{The general case being too hard to analyse, we found the constant rate enlightening for understanding how the algebraic attacks are affected by going from MCEP to MSEP}.

Indeed, in MSE\new{P} the $k(nm-k)$ bihomogeneous equations in the $n^2+m^2$ $A$ and $B$ variables are missing. Assuming $k=Rnm$ for some $R$ in $[0,1]$ and $m=n$, we end up with $(R-R^2)n^4$ equations in $n^2+n^2$ variables. In order to linearise, the number of equations has to be compared to the number of monomials, $n^4$ in this bihomogeneous case. In bi-degree $(1,1)$ linearisation is impossible since $(R-R^2)n^4 < n^4$. Taking $d>1$ and looking in bi-degree $(d,1)$, we have $(R-R^2)\binom{n^2+d-2}{d}n^4$ equations for $\binom{n^2+d-1}{d}n^2$ monomials. The case of bi-degree $(1,d)$ is symmetrical. The ratio equations over monomials is $\frac{(R-R^2)d}{1-\frac{d-1}{n^2}}$. \old{For $n>2$, $R=1/2$ and $d\geq 6 > 1+\frac{3}{1-\frac{4}{n^2}}$, there are more equations than monomials.} \new{This implies that for $n>1$, whenever $d\geq \frac{2}{R(1-R)}$, there are more equations than monomials.} Here $d$ is independant from $n$. Experimentally, no syzygy arises before $d=n$ and linearisation is possible \new{in MCEP}.

For the set of bihomogeneous equations in $A,T$ ($B,T$ resp.) variables existing in MSE\new{P}, there are $n(mk'-n)$ ($m(nk'-m)$ resp.) equations in $m^2+kk'$  ($n^2+kk'$ resp.) variables and in constant rate $(k'-1)n^2$ equations for $n^2+Rk'n^2$ variables. In bi-degree $(d,1)$, there are $(k'-1)n^2\binom{n^2+d-2}{d-1}$ equations for $\binom{n^2+d-1}{d}Rk'n^2$ monomials intervening. In order to linearise, we need the ratio $\frac{(k'-1)n^2\binom{n^2+d-2}{d-1}}{\binom{n^2+d-1}{d}Rk'n^2} >1 \Leftrightarrow d > \frac{R(n^2-1)}{1-R-\frac{1}{k'}}$. Thus we need $d=\Theta(n^2)$ with no other constraint on $k'$ than $\frac{1}{k'}\ne 1-R$. Moreover, at such a big degree $d$, syzygies can arise letting even fewer equations to linearise. The case of bi-degree $(1,d)$ isn't symmetrical but analogous, we also end up with a constraint $d=\Theta(n^2)$. One could try to turn a MSE\new{P} into a MCE\new{P} by guessing some rows of $T$ which would leave a square matrix of variables assumed to be invertible. For $k'=k-1$, the attacker would have to guess $q^{k-1}$ values which is already exponential in $n$ when $k=\Theta(n)$. For now, MSE\new{P} seems to be already far harder than MCE\new{P} at constant rate. However, the equations $0 = \bm{G'_T}(\bm{A}^{-1}\otimes\bm{B}^{-\top})\bm{{H}_T}^\top$ we ignored in the trilinear form modeling (since they added $m^2+n^2$ new variables) are now of use and by themselves, enable to linearise for some MSE\new{P} instances. Indeed, we have $k'(nm-k)=k'(1-R)n^2$ equations in bi-degree $(1,1)$ for $n^2+m^2$ variables. As previously, looking at the ratio between the number of equations in bi-degree $(d,1)$ over the number of monomials  $\frac{k'(1-R)n^2\binom{m^2+d-2}{d-1}}{n^2\binom{m^2+d-1}{d}}$  gives the following constraint on $d$ for linearisation in $(d,1)$: $d> \frac{n^2-1}{k'(1-R)-1}$. \new{In order to have $d=\Theta(1)$, it would need $k'=\Theta(n^2)=\Theta(k)$.}

This example of constant rate strongly emphasizes the lack of equations in MSE\new{P} compared to MCE\new{P} when $k'=\Theta(1)$.

\subsection{Leon algorithm}
\subsubsection{Two-collision approach}
In Hamming metric, the main algorithm to solve the equivalence problem is Leon's algorithm \cite{Leon}. The algorithm aims at building the codewords of small weight, and from them, recovering the permutation. This algorithm works since the permutation does not change the weight of the words. In the Matrix Code Equivalence problem, this algorithm is adapted to the rank metric \cite{MEDS1}, where finding codewords of weight $r$ corresponds to a $(q,m,n,k,r)$ MinRank instance. On parameters such as what is used for $\MSE$, the complexity of the MinRank instance is typically not very high, and will thus not be the determining factor. In the case of $\MCE$, the weight distribution between the two codes is identical and as a result the numbers of codewords generated from $\C$ and from $\D$ are the same. In the subcode equivalence however, the minimum weight of the two codes will not be the same. The lists will thus be unbalanced, which will lead to a difference in the complexity. In practice, we apply Algorithm \ref{LeonAlg}, which is the same attack as in \cite{MEDS1} and \cite{MEDS2}.
\begin{algorithm}[h]
\caption{\old{Adaptation of Leon's algorithm for an $\MSE$ instance} \new{$\MSE$ solver based on Leon's algorithm using two collisions}}
\label{LeonAlg}
\begin{algorithmic}[1]
\REQUIRE An $\MSE$ instance $\bm{G} \in \Fq^{k \times mn},\bm{G}'\in \Fq^{k' \times mn}$. 
\ENSURE If they exist, returns $\bm{A} \in \GL_m(q)$, $\bm{B} \in \GL_n(q)$ such that $\bm{G'}=\bm{T}^\top\bm{G}(\bm{A}^\top\otimes\bm{B})$ for some $\bm{T}\in\Fq^{k\times k'}$ of rank $k'$. Otherwise, returns $\perp$.

\STATE Build two lists $\mathcal{L}_1$ and $\mathcal{L}_2$, of size $N_1$ and $N_2$, of matrices from $\C$ and $\D$ respectively, of rank $r$.
\STATE Set $\Tilde{\bm{A}} \in \Fq[\bm{a}]^{m \times m}$, $\Tilde{\bm{B}} \in \Fq[\bm{b}]^{n \times n}$, $\bm{A}' \in \Fq[\bm{a}']^{m \times m}$, $\bm{B}' \in \Fq[\bm{b}']^{n \times n}$ where $\bm{a},\bm{b},\bm{a}',\bm{b}'$ are sets of unknowns.
\REPEAT
\STATE $\bm{C}_1 \sampler \mathcal{L}_1$ and $\bm{C}_2 \sampler \mathcal{L}_1$.
\STATE $\bm{D}_1 \sampler \mathcal{L}_2$ and $\bm{D}_2 \sampler \mathcal{L}_2$.
\STATE Use the system $\label{eqleon}
    \begin{cases}
      \Tilde{\bm{A}}\bm{C}_1 = \bm{D}_1\bm{B}'\\
      \Tilde{\bm{A}}\bm{C}_2 = \bm{D}_2\bm{B}'\\
      \bm{A}'\bm{D}_1 = \bm{C}_1\Tilde{\bm{B}}\\
      \bm{A}'\bm{D}_2 = \bm{C}_2\Tilde{\bm{B}}\\
    \end{cases}$ to express $\Tilde{\bm{A}}$ and $\Tilde{\bm{B}}$ with fewer variables.

\STATE Use $\Tilde{\bm{A}}$ and $\Tilde{\bm{B}}$ to apply one of the algebraic attacks with fewer variables, which returns two matrices $\bm{A} \in \Fq^{m\times m}$ and $\bm{B} \in \Fq^{n \times n}$.
\UNTIL{$\mathsf{rank}\left(\begin{matrix}\bm{G} \\ \bm{G}'\big(\bm{A}^\top \otimes \bm{B}) \end{matrix} \right) = k$ and $\bm{A} \in \GL_m(q)$ and $\bm{B} \in \GL_n(q)$ or all tuples $(\bm{C}_1,\bm{C}_2,\bm{D}_1,\bm{D}_2)$ have been used}
\IF{$\mathsf{rank}\left(\begin{matrix}\bm{G} \\ \bm{G}'\big(\bm{A}^\top \otimes \bm{B}) \end{matrix} \right) = k$ and $\bm{A} \in \GL_m(q)$ and $\bm{B} \in \GL_n(q)$}
\STATE Return $\bm{A},\bm{B}$.
\ENDIF
\STATE Return $\perp$
\end{algorithmic} 
\end{algorithm}

The complexity is not straightforward however. The size of the list must be studied, and the solving of the system must be looked at as well.

\paragraph{Solving of the system.} Since the matrices $\bm{C}_1,\bm{C}_2,\bm{D}_1,\bm{D}_2$ are not of full rank (they have rank $r$), we will not have enough independent equations to fully solve. One can obtain this way at most $2(mn-(n-r)(m-r))$ equations in $\Tilde{\bm{A}}$ and $\bm{B}'$, and the same number in $\bm{A}'$ and $\Tilde{\bm{B}}$. To explain this number of linear equations, consider the first line of the system:
$$\Tilde{\bm{A}}\bm{C}_1 = \bm{D}_1\bm{B}' $$ where $\bm{C}_1 $ and $\bm{D}_1$ are two rank $r$ matrices. When putting $\bm{C}_1$ and $\bm{D}_1$ under row-reduced echelon form (for $\bm{D}_1$) and column-reduced echelon form (for $\bm{C}_1$), which corresponds to a multiplication by invertible matrices $\bm{D}_{\mathsf{ech}}$ on the left and $\bm{C}_{\mathsf{ech}}$ on the right, we obtain $$\bm{D}_{\mathsf{ech}}\Tilde{\bm{A}}\bm{C}_1\bm{C}_{\mathsf{ech}} = \bm{D}_{\mathsf{ech}}\bm{D}_1\bm{B}'\bm{C}_{\mathsf{ech}} $$
where $\bm{C}_1\bm{C}_{\mathsf{ech}} = \begin{bmatrix}\bm{W} & \bm{0} \end{bmatrix}$ and $\bm{D}_{\mathsf{ech}}\bm{D}_1 = \begin{bmatrix} \bm{V} \\ \bm{0} \end{bmatrix}$ for some matrices $\bm{W} \in \Fq^{m \times r}$ and $\bm{V} \in \Fq^{r \times n}$. We then obtain, after putting all the terms on one side, $$ \begin{bmatrix}\bm{W}' & \bm{0} \end{bmatrix} - \begin{bmatrix} \bm{V}' \\ \bm{0} \end{bmatrix} = \bm{0}$$ for some matrices $\bm{W}' \in \Fq[\bm{a}]^{m \times r}$ and $\bm{V}' \in \Fq[\bm{b}']^{r \times n}$ where $\bm{a}$ and $\bm{b}'$ are the unknowns of $\Tilde{\bm{A}}$ and $\bm{B}'$, such that $ \begin{bmatrix}\bm{W}' & \bm{0} \end{bmatrix} $ $ \in \Fq[\bm{a}]^{m \times n}$ and   $\begin{bmatrix} \bm{V}' \\ \bm{0} \end{bmatrix} $ $ \in \Fq[\bm{b}']^{m \times n}$. Because of the submatrix of $\bm{0}$ in both matrices, this results in at most $mn-(n-r)(m-r)$ independent linear equations. Considering first the collision on $\bm{C}_1$ and $\bm{D}_1$ thus gives $mn-(n-r)(m-r)$ equations, and then doing the same on $\bm{C}_2$ and $\bm{D}_2$ gives $mn-(n-r)(m-r)$ at most as well. When considering the two collisions, we have at most  $2(mn-(n-r)(m-r))=2r(m+n-r)$ independent equations.

Due to the rank of the matrices, it is obviously not possible to fully solve the system. However, there are enough equations to obtain a basis of the space of solutions of $\bm{A}$ and $\bm{B}^{-1}$ thanks to the two first lines of the system (i.e, the $mn-(n-r)(m-r)$ first equations), and a basis of the space of $\bm{A}^{-1}$ and $\bm{B}$ thanks to the two last lines (i.e, the $mn-(n-r)(m-r)$ last equations). It turns out one can find a solution space of $\bm{A}$ that does not depend on the variables from $\bm{B}^{-1}$, and the same goes when the considered variables are $\bm{B}$ and $\bm{A}^{-1}$. Thus, we obtain a space of solutions of $\bm{A}$ and $\bm{B}$ of dimension $(m-n)^2+2(m-r)(n-r)$. This reduces the number of variables that intervene in any of the algebraic modeling, making them easier to solve. \new{This also allows one to express the variables in $\Tilde{\bm{A}}$ as variables $\bm{B'}$ (and reciprocally), thus obtaining equations $\bm{AA}^{-1}$ and $\bm{BB}^{-1}$ (in the algebraic modelings), resulting in an even lower complexity for the polynomial system solving.} However, the extent to which it reduces their complexity is still to be determined \new{and we will thus disregard the complexity of solving the system in our estimates as it would be possible to solve the system much faster than for a generic polynomial system.}.

\paragraph{The sizes of the lists.} To obtain the complexity of the attack, one needs to know the sizes of the lists, i.e., $N_1$ and $N_2$. First, we will note $C(r)$ and $C'(r)$ the expected number of codewords of rank $r$ in $\C$ and $\D$ respectively. 
For that, we recall the well-known lemma:
\begin{lemma}\label{nbmat}
    The number of matrices in $\Fq^{m \times n}$ of rank $r$ is
    $$M_{m,n}(r)= \prod_{i=0}^{r-1}\frac{(q^m-q^i)(q^{n}-q^i)}{q^r-q^i} $$
\end{lemma}

Then, the expected number $C(r)$ of codewords of rank $r$ is simply the probability that a uniformly sampled matrix is of rank $r$, multiplied by the number of matrices in the code, i.e. $C(r)=M_{m,n}(r)/q^{mn-k}$. Since $\D$ is a subcode (and its dimension is small in practice), we can consider that one will build all the codewords of rank $r$ in $\D$, i.e, $N_2 = C'(r)=M_{m,n}(r)/q^{mn-k'}$. Note that it is possible to build fewer codewords, which we will detail further. Then, to obtain \underline{one} collision, one can either:
\begin{itemize}
    \item sample $C(r) -N_2+1$ codewords of rank $r$ in $\C$. This assures that at least one collision is found;
    \item or sample $N_1$ codewords of rank $r$ in $\C$ with $N_1 \le C(r)-N_2$. We detail this case below.
\end{itemize}
Let $N_1 \le C(r)-N_2$. Then, one can obtain a collision with probability $c$ if  \[\frac{\binom{C(r)-N_2}{N_1}}{\binom{C(r)}{N_1}} = 1-c.\]

This equation arises from the fact that the number of ways to sample $N_1$ elements from $C(r)$ possibilities is $\binom{C(r)}{N_1}$, and there are $\binom{C(r)-N_2}{N_1}$ possible samplings where no collision is found.

Then, since $$\frac{\binom{C(r)-N_2}{N_1}}{\binom{C(r)}{N_1}} = \frac{(C(r)-N_2) \times \dots \times (C(r)-N_2-N_1+1)}{C(r)\times \dots \times (C(r)-N_1+1)} \ge \Big(\frac{C(r) - C'(r)}{C(r)}\Big)^{N_1}$$ and by taking $N_2 = C'(r)$, it is possible to get a lower bound of $N_1$, which is $N_1$ such that \[ \Big(\frac{C(r) - C'(r)}{C(r)}\Big)^{N_1} \le 1-c\]
Finally, we obtain that $N_1$ must be such that \[N_1 \ge \frac{\log(1-c)}{\log(\frac{C(r)-C'(r)}{C(r)})}.\]
Using the fact that $\log(1-x)\le -x$ when $x\in[0..1]$ (here, $x = \frac{C'(r)}{C(r)}$, where $C(r)$ is much larger than $C'(r)$), we easily obtain that 

$$N_1 \ge \frac{- \log(1-c) \cdot C(r)}{C'(r)}$$
and, by using the actual values of $C(r)=M_{m,n}(r)/q^{mn-k}$ and $C'(r)=M_{m,n}(r)/q^{mn-k'}$, we obtain that $$N_1\ge -\log(1-c) \cdot q^{k-k'}.$$ 
Then, since one has to compute all the possible pairs to obtain one collision, the complexity for one collision is $N_1 N_2 \mathsf{C_{SolveAlg}}$, as for each element in $\mathcal{L}_1\times \mathcal{L}_2$ we need to check if it is a real collision by solving the system, where $\mathsf{C_{SolveAlg}}$ is the complexity to solve the system given by any algebraic attack using the space of the solutions of $\bm{A}$ and $\bm{B}$. Since we want $2$ collisions, $N_1$ needs to be slightly higher (although negligible compared to $q^{k-k'}$), and we now need to check the compatibility of the system for each pair of elements in $\mathcal{L}_1\times \mathcal{L}_2$, resulting in a complexity proportional to $\binom{N_1N_2}2$. 
From this analysis, we obtain the following result:
\begin{lemma}
    Algorithm \ref{LeonAlg} finds a solution to the $\MSE$ instance in time at least $$ \mathcal{O}\left( \binom{N_1\cdot N_2}{2} (mn)^\omega \mathsf{C_{SolveAlg}} \right), \text{~~~~} n,k \rightarrow \infty $$ where $r$ is taken such that matrices of rank $r$ exists in the subcode $\D$. 
\end{lemma} 
This is the same result obtained as for $\MCE$: in all the estimates, we will take $N_2 = C'(r)$ and $N_1=q^{k-k'}$, meaning $N_1N_2 = C(r)$, which matches \cite{MEDS1}. 
However, due to the low dimension of the subcode, the weight repartition of the matrices changes: there will be much fewer matrices that are not of full rank. Compared to $\MCE$, this implies that the rank $r$ that must be taken in the algorithm will be higher. \new{Furthermore, as $\mathsf{C_{SolveAlg}}$ is unknown at the moment, we will only consider the complexity of finding two collisions to establish our parameters as it is a lower bound of the complexity.} Finally, we consider that the attack is possible as long as $C'(r) \ge 2^{-\lambda}$, as one (or several) codeword of rank $r$ could exist in the subcode for some instances of the $\MSE$. 

One should also note that it is possible to improve the attack by considering a simple fact: there are $(q-1)^2$ couples $(\bm{A},\bm{B})$ solutions, because $\bm{A}$ and $\lambda \bm{A}$ (resp. $\bm{B}$ and $\lambda \bm{B}$) for $\lambda \in \Fq$ are both solutions. This means that it is possible to reduce the size of $C'(r)$ by a factor $(q-1)$ (by removing the matrices which are the same up to a scalar) and the same goes for $C(r)$ (because then we have just one collision with smaller lists). Overall, this reduces the cost of the attack by a factor $(q-1)^2$.

\subsubsection{One-collision approach}
It is possible to improve the above attack, by using the one-collision approach from \cite{MEDS2}. The idea is to use only one collision between two codewords $\bm{C}_1$ and $\bm{D}_1$. Having a collision between the two codewords means $\bm{A}\bm{C}_1\bm{B}=\bm{D}_1$. The idea from \cite{MEDS2} is then to use the kernel of $\bm{C}_1$ and $\bm{D}_1$: since $\bm{C}_1$ and $\bm{D}_1$ are of rank $r$, their kernels are of dimension $(n-r)$ (resp. $(m-r)$ for the left kernel). Let $\bm{K} \in \Fq^{n \times (n-r)}$ be the matrix of the right kernel of $\bm{D}_1$. Then, we have that $\bm{A}\bm{C}_1\bm{B}\bm{K}=\bm{0}$. This means that $\bm{C}_1\bm{B}\bm{K} = \bm{0}$ (since $\bm{A}$ is invertible). This system possesses $r(n-r)$ independent equations (as $\Rank(\bm{C}_1)=r$) in  $n^2$ unknowns: compared to the previous modeling, this removed the variables in $\bm{A}^{-1}$, which is a better approach due to the lack of the second collision. This yields $r(n-r)$ linear equations in the variables of $\bm{B}$ and $r(m-r)$ linear equations in the variables of $\bm{A}$ by applying the same process to the left. As a result, we can reduce the number of variables of $\bm{A}$ to $m^2-rm+r^2$, and the number of variables of $\bm{B}$ to $n^2-nr+r^2$. It is then possible to apply the algebraic modeling with fewer variables to obtain the solution.  We sum up this attack in Algorithm \ref{alg:onecoll}.

\begin{algorithm}[h]
\caption{\old{Adaptation of Leon's algorithm for an $\MSE$ instance} \new{$\MSE$ solver based on Leon's algorithm using one collision}}
\label{alg:onecoll}
\begin{algorithmic}[1]
\REQUIRE An $\MSE$ instance $\bm{G} \in \Fq^{k \times mn},\bm{G}'\in \Fq^{k' \times mn}$.
\ENSURE If they exist, returns $\bm{A} \in \GL_m(q)$, $\bm{B} \in \GL_n(q)$ such that $\bm{G'}=\bm{T}^\top\bm{G}(\bm{A}^\top\otimes\bm{B})$ for some $\bm{T}\in\Fq^{k\times k'}$ of rank $k'$. Otherwise, returns $\perp$.
\STATE Build two lists $\mathcal{L}_1$ and $\mathcal{L}_2$, of size $N_1$ and $N_2$, of matrices from $\C$ and $\D$ respectively, of rank $r$.
\STATE Set $\Tilde{\bm{A}} \in \Fq[\bm{a}]^{m \times m}$, $\Tilde{\bm{B}} \in \Fq[\bm{b}]^{n \times n}$ where $\bm{a},\bm{b}$ are sets of unknowns.
\REPEAT
\STATE $\bm{C}_1 \sampler \mathcal{L}_1$.
\STATE $\bm{D}_1 \sampler \mathcal{L}_2$.
\STATE Compute $\bm{K} \in \Fq^{n \times (n-r)}$ of rank $n-r$ such that $\bm{D}_1\bm{K} = \bm{0}$.
\STATE Compute $\bm{K}' \in \Fq^{(m-r) \times m}$ of rank $m-r$ such that $\bm{K}'\bm{C}_1 = \bm{0}$.
\STATE Use the system $\begin{cases}\bm{C}_1\Tilde{\bm{B}}\bm{K} = \bm{0}\\\bm{K}'\Tilde{\bm{A}}\bm{C}_1 = \bm{0}\end{cases}$ to express $\Tilde{\bm{A}}$ and $\Tilde{\bm{B}}$ with fewer variables.
\STATE Use $\Tilde{\bm{A}}$ and $\Tilde{\bm{B}}$ to apply one of the algebraic attacks with fewer variables, which returns two matrices $\bm{A} \in \Fq^{m\times m}$ and $\bm{B} \in \Fq^{n \times n}$.
\UNTIL{$\mathsf{rank}\left(\begin{matrix}\bm{G} \\ \bm{G}'\big(\bm{A}^\top \otimes \bm{B}) \end{matrix} \right) = k$ and $\bm{A} \in \GL_m(q)$ and $\bm{B} \in \GL_n(q)$ or all tuples $(\bm{C}_1,\bm{D}_1)$ have been used}
\IF{$\mathsf{rank}\left(\begin{matrix}\bm{G} \\ \bm{G}'\big(\bm{A}^\top \otimes \bm{B}) \end{matrix} \right) = k$ and $\bm{A} \in \GL_m(q)$ and $\bm{B} \in \GL_n(q)$}
\STATE Return $\bm{A},\bm{B}$.
\ENDIF
\STATE Return $\perp$
\end{algorithmic} 
\end{algorithm}

From the previous analysis, we obtain the following lemma:

\begin{lemma}The Algorithm \ref{alg:onecoll} finds a solution to the $\MSE$ instance in time at least $$ \mathcal{O} \left( N_1N_2(mn)^\omega\mathsf{C_{solve}} \right), \text{~~~~} n,k \rightarrow \infty$$ where $\mathsf{C_{solve}}$ is the cost of the algebraic modeling with fewer variables, and where $r$ is taken such that matrices of rank $r$ exists in the subcode $\D$. \end{lemma} In the estimate, $r$ is taken such that $C'(r) \ge 2^{-\lambda}$ (as only one codeword of rank $r$ would be enough to perform this attack) and the cost of the attack is minimal. Once again, this is the same complexity that is given in \cite{MEDS2}, but $r$ is higher than $\MCE$ in practice, as codewords have a higher weight on average in a subcode.

\subsection{\old{QMLE} \new{Attacks based on Quadratic Map Linear Equivalence solver}}
In \cite{RST22}, the authors described an attack to solve $\MCE$ using the fact that it is equivalent to the \new{Quadratic Map Linear Equivalence ($\QMLE$)} problem, which we recall below:
\begin{problem}\label{QMLE} $\QMLE(q,k,N,\mathcal{F},\mathcal{P})$:\\
Let $\F = (f_1, \dots, f_k) \in \Fq[x_1,\dots,x_N]^k$ and $\Pcal = (p_1, \dots, p_k)$ $\in \Fq[x_1,\dots,x_N]^k$. Determine if there exist $\bm{S} \in \GL_N(q)$ and $\bm{T}\in \GL_k(q)$ such that $$\Pcal(\bm{x})=\F(\bm{xS})\bm{T}$$
\end{problem}

In the case of the subcode however, the problem that interests us is when $\F$ and $\Pcal$ are not both in $\Fq[x_1,\dots,x_N]^k$. In fact, $\Pcal$ will be in $\Fq[x_1,\dots,x_N]^{k'}$, since $\bm{T}\in \Fq^{k \times k'}$: 

This problem seems to be halfway between the Isomorphism of polynomials problem and the Morphism of Polynomials problem, where one requires both matrices invertible and the other does not require any of them to be invertible. We will call this problem $\QSMLE$, for Quadratic Sub Map Linear Equivalence.

\begin{problem}\label{QSMLE} $\QSMLE(q,k,k',N,\mathcal{F},\mathcal{P})$:\\
Let $\F = (f_1, \dots, f_k) \in \Fq[x_1,\dots,x_N]^k$ and $\Pcal = (p_1, \dots, p_k')$ $\in \Fq[x_1,\dots,x_N]^{k'}$. Determine if there exist $\bm{S} \in \GL_N(q)$ and $\bm{T}\in \Fq^{k\times k'}$ of rank $k'$ such that $$\Pcal(\bm{x})=\F(\bm{xS})\bm{T}$$
\end{problem}

\paragraph{From $\MSE$ to $\QSMLE$.}
As for $\MCE$, we can obtain a $\QSMLE$ instance from an $\MSE$ one.
\begin{lemma}
    Let $q,m,n,k,k'$ be positive integers with $k' < k$. Let $\mathcal{A}$ be an algorithm which solves a $(q,k,k',m+n)$ $\QSMLE$ instance with probability $\epsilon_{1}$. Then, there is an algorithm $\mathcal{A}'$ that solves a $(q,m,n,k,k')$ $\MSE$ instance with probability $\epsilon_{2}$ where $\epsilon_2 \ge \epsilon_1$
\end{lemma}
\begin{proof}
To reduce an instance from $\MSE$ to $\QSMLE$, one proceeds exactly as in \cite{RST22}, by setting $\bm{x}=\begin{pmatrix*}x_1&\ldots & x_n &x_{m+1} & \ldots & x_{n+m}\end{pmatrix*}$ and
\begin{center}
    \begin{align*}
        p_{i}(\bm{x}) & = \begin{pmatrix*}x_1 & \ldots & x_m\end{pmatrix*}\bm{D}^{(i)}\begin{pmatrix*}x_{m+1} & \ldots & x_{m+n}\end{pmatrix*}^\top \\
        f_i(\bm{x}) & = \begin{pmatrix*}x_1 & \ldots & x_m\end{pmatrix*}\bm{C}^{(i)}\begin{pmatrix*}x_{m+1} & \ldots & x_{m+n}\end{pmatrix*}^\top
    \end{align*}
\end{center}

 where $\D = \langle \bm{D}^{(1)},\dots,\bm{D}^{(k')} \rangle$, and where $\C = \langle \bm{C}^{(1)},\dots,\bm{C}^{(k)} \rangle$ are the codes of the $\MSE$ instance.
 By definition of the $\MSE$, the solution of the $\QSMLE$ instance will be of the form $\bm{S} = \begin{bmatrix}
   \bm{A}  & \bm{0} \\
    \bm{0} & \bm{B}^\top 
\end{bmatrix}$. Thus, solving the $\QSMLE$ instance directly yields the solutions of the $\MSE$ one.
\end{proof}

We will denote with $\iQSMLE$ the \textit{inhomogeneous} variant of $\QSMLE$ and $\hQSMLE$ and \textit{homogeneous} variant, i.e.,

\begin{problem}\label{iQSMLE} $\iQSMLE(q,k,k',N,\mathcal{F},\mathcal{P})$:\\
Let $\F = (f_1, \dots, f_k) \in \Fq[x_1,\dots,x_N]^k$ and $\Pcal = (p_1, \dots, p_k')$ $\in \Fq[x_1,\dots,x_N]^{k'}$ two tuples of inhomogeneous polynomials. Determine if there exist $\bm{S} \in \GL_N(q)$ and $\bm{T}\in \Fq^{k\times k'}$ of rank $k'$ such that $$\Pcal(\bm{x})=\F(\bm{xS})\bm{T}$$
\end{problem}

\begin{problem}\label{hQSMLE} $\hQSMLE(q,k,k',N,\mathcal{F},\mathcal{P})$:\\
Let $\F = (f_1, \dots, f_k) \in \Fq[x_1,\dots,x_N]^k$ and $\Pcal = (p_1, \dots, p_k')$ $\in \Fq[x_1,\dots,x_N]^{k'}$ two tuples of homogeneous polynomials. Determine if there exist $\bm{S} \in \GL_N(q)$ and $\bm{T}\in \Fq^{k\times k'}$ of rank $k'$ such that $$\Pcal(\bm{x})=\F(\bm{xS})\bm{T}$$
\end{problem}

We will see later that the inhomogeneous instances are easier than the homogeneous ones. It is thus beneficial to attack the problem by transforming a homogeneous instance into an inhomogeneous one.

\paragraph{General solver for $\QSMLE$.} We now recall the principle of the first attack from \cite{RST22}:
\begin{itemize}
    \item Build two lists $\mathcal{L}_1,\mathcal{L}_2\subseteq \Fq^{n+m}$ of size $N_1$ and $N_2$ respectively;
    \item For each pair of elements $(\bm{a},\bm{b}) \in \mathcal{L}_1 \times \mathcal{L}_2$, compute $\F(\bm{x}+\bm{a})$ and $\Pcal(\bm{x}+\bm{b})$, and find an isometry between them by solving $\iQSMLE$;
    \item Check if the isometry corresponds to a solution of the $\MSE$ instance.
\end{itemize}
If one manages to find $(\bm{a},\bm{b})$ such that $\bm{b} = \bm{aS}$, then $\F((\bm{x}+\bm{a})\bm{S})\bm{T}=\F(\bm{x}\bm{S}+\bm{aS})\bm{T} = \Pcal(\bm{x}+\bm{aS})=\Pcal(\bm{x}+\bm{b})$. To find such $(\bm{a},\bm{b})$, one can sample one $\bm{a} \in \Fq^{m+n}$, and then solve the inhomogeneous instance for every $\bm{b} \in \Fq^{m+n}$ (using a solver $ \mathsf{Solve}_{\mathsf{\iQSMLE}}$). The inhomogeneous instance will be solvable if $\bm{b} = \bm{aS}$. We sum up this process in Algorithm \ref{alg:QMLEalg}.

To reduce the cost of the attacks, \cite{bouillaguet} and \cite{RST22} use $\mathcal{D}_{\bm{a}}(\F) = \F(\bm{x}+\bm{a})-\F(\bm{x})-\F(\bm{a})$, which is a linear application (it is easy to see that $\mathcal{D}_a(f_i)(\bm{x}) = (x_1,\dots,x_m)\bm{C}^{(i)}(a_{m+1}\dots a_{m+n})^\top+ (a_1,\dots,a_m)\bm{C}^{(i)}(x_{m+1}\dots x_{m+n})^\top$, which is linear in $\bm{x})$). Then, they look at the dimensions of the kernel of this linear application and keep only the elements $\bm{a}$ such that $\dim(\mathsf{ker}(\mathcal{D}_{\bm{a}}(\F))) = 1$, and proceeding in the same way with $\Pcal$, which reduces the sizes of the lists. Adapting this invariant is not possible for $\QSMLE$, as the dimension of the kernel $\mathcal{D}_{\bm{a}}(\F)$ is not an invariant, as $\F$ is multiplied by a non-invertible matrix.

\begin{algorithm}[h!]
\caption{General solver for the $\QSMLE$ problem}
\label{alg:QMLEalg}
\begin{algorithmic}[1]
\REQUIRE An $\MSE$ instance $\bm{G} \in \Fq^{k \times mn},\bm{G}'\in \Fq^{k' \times mn}$. 
\ENSURE $\bm{A} \in \GL_m(q)$, $\bm{B} \in \GL_n(q)$.
\STATE Set $\bm{A} \leftarrow \bm{0}$ and $\bm{B} \leftarrow \bm{0}$.
\STATE $\bm{a} \sampler \Fq^{m+n}$.
\STATE $E = \{ \}$.
\REPEAT
\STATE $\bm{b} \sampler \Fq^{m+n} \setminus{E}$.
\STATE $E = E \cup {\bm{b}}$.
\STATE $(\bm{A},\bm{B}) \sampler \mathsf{Solve}_{\mathsf{\iQSMLE}}\left( \F(\bm{x}+\bm{a}), \Pcal(\bm{x}+\bm{b})\right)$.
\UNTIL{$\mathsf{rank}\left(\begin{matrix}\bm{G} \\ \bm{G}'\big(\bm{A}^\top \otimes \bm{B}) \end{matrix} \right) = k$ and $\bm{A} \in \GL_m(q)$ and $\bm{B} \in \GL_n(q)$}
\STATE Return $\bm{A},\bm{B}$.
\end{algorithmic} 
\end{algorithm}

Since there are $(q-1)^2$ values of $\bm{S}$ possible (we remind that there are $(q-1)^2$ solutions to the $\MSE$ instance, because multiplication by a scalar keeps the solution valid), one needs to test around only $q^{m+n-2}$ values of $\bm{b}$. We thus obtain the following lemma:
\begin{lemma}
    The Algorithm \ref{alg:QMLEalg} finds a solution to the $\MSE$ in time at least $$(q^{m+n-2}-n_{\mathsf{sol}}) \mathsf{C}_{\iQSMLE} +  n_{\mathsf{sol}}\mathsf{C}_{\iQSMLE} (mn)^\omega $$ where $n_{\mathsf{sol}}$ is the number of queries for which $\mathsf{Solve}_{\mathsf{\iQSMLE}}\left( \F(\bm{x}+\bm{a}), \Pcal(\bm{x}+\bm{b}) \right)$ returns $(\bm{A},\bm{B})$ instead of $\perp$, and $\mathsf{C}_{\iQSMLE}$ is the cost of running the $\mathsf{Solve}_{\mathsf{\iQSMLE}}$ algorithm.
\end{lemma}

\begin{remark}
In the above complexity, $n_{\mathsf{sol}}$ is expected to be only $(q-1)^2$, as this corresponds to the number of solutions $(\bm{A},\bm{B})$. This can be explained mainly by the fact that there is no reason for the inhomogeneous instance to have a solution if $\bm{b} \ne \bm{aS}$. Thus, very few ranks of matrices actually need to be computed. 
\end{remark}

\paragraph{Bilinear solver.} As to the second algorithm proposed in \cite{RST22}, it aims at matching elements where the kernel is of positive dimension. The difference is that elements are built by sampling them in $\Fq^{m}$ instead of $\Fq^{m+n}$, making direct use of the fact that $\F$ and $\Pcal$ are $k$ (or $k'$)-dimensional bilinear forms. In what follows, we will thus consider either $\F(\bm{x})$ with $\bm{x} \in \Fq^{m+n}$ or $\F(\bm{x},\bm{y})$ with $\bm{x} \in \Fq^{m}$ and $\bm{y} \in \Fq^{n}$ (and the same for $\Pcal$). The algorithm uses the fact that $\Pcal(\cdot,\bm{b})$ is linear for any $\bm{b} \in \Fq^{n}$, making it easy to compute its kernel. 

For $\MCE$, this consideration results in a complexity in $q^{\min(m,n,k)}$. However, in the case of $\MSE$, once again, this complexity changes greatly due to the subcode. In fact, we have a different lemma than \cite[Lemma 35 4-5]{RST22}: there is no equality between $\operatorname{Ker}\left( \F({\bm{aA}},\bm{y} \bm{B}^\top) \right)$ and $\operatorname{Ker}\left( \Pcal(\bm{a},\bm{y}) \right)$. In particular, we have the following lemma:

\begin{lemma}
    Let $\bm{a} \in \Fq^m$, and $(\bm{A},\bm{B},\bm{T},\F,\Pcal)$ be from a $\QSMLE$ instance. Then, $\operatorname{Ker}\left( \F({\bm{aA}},\bm{y} \bm{B}^\top) \right) \subset \operatorname{Ker}\left( \Pcal(\bm{a},\bm{y}) \right)$
\end{lemma}

\begin{proof}
We recall that we have that by definition of $\QSMLE$, $\Pcal(\bm{a},\bm{y})=\F(\bm{aA},\bm{yB}^\top)\bm{T}$ where $\bm{T} \in \Fq^{k\times k'}$. 
It is then straightforward that $\F(\bm{aA},\bm{yB^\top}) = \bm{0}$ for some fixed $\bm{a} \in \Fq^{m}$ does imply $\Pcal(\bm{a},\bm{y}) = \bm{0}$. However,  the reciprocal is not true as $\Pcal(\bm{a},\bm{y}) = \F(\bm{aA},\bm{yB^\top})\bm{T}$ and $\bm{T}$ is not invertible.
\end{proof}

This lemma implies that to obtain a working algorithm similar to \cite{RST22}, matching the elements will be more difficult. In fact, this strict inclusion means it is necessary to sample elements using one element of the kernel of $\F$, and all the elements from the one from $\Pcal$.

More precisely, the algorithm asks to:
\begin{itemize}
    \item Build a set $\Tilde{P}_b = \{ \bm{b} \in \Fq^{n} ~|~ \operatorname{dim}\operatorname{Ker} \Pcal(\cdot,\bm{b}) > 0  \}$;
    \item Build the set $\Tilde{P}_a =\{\bm{a} \in \Fq^{m} ~|~ \operatorname{dim}\operatorname{Ker} \Pcal(\bm{a},\cdot) > 0  \} $ from the set $\Tilde{P}_b$ (because these two sets are related, see \cite{RST22});
    \item Build one element $(\bm{a},\bm{b})$ in $\Tilde{F}_a \times \Tilde{F}_b$ (where these sets are defined similarly). Then, sample an element $(\bm{a}',\bm{b}')$ in $\Tilde{P}_a \times \Tilde{P}_b$ and solve the $\iQSMLE$ instance $\left( \F(\bm{x}+\bm{a},\bm{y}+\bm{b}), \Pcal(\bm{x}+\bm{a}',\bm{y}+\bm{b}')\right)$ until a solution is found.
\end{itemize}

The complexity is then determined only by the expected cardinal of the sets: $|\Tilde{F}_a| = \frac{M_{k,n}(n-1)q^m}{q^{kn}}$, $|\Tilde{F}_b| = \frac{M_{k,m}(m-1)q^n}{q^{km}}$, $|\Tilde{P}_a| = q^{m}$, $|\Tilde{P}_b| = q^{n}$ (as $k' < n$, $\operatorname{Ker} \Pcal(\cdot,\bm{b})$ is always non-empty, and the same for $\operatorname{Ker} \Pcal(\bm{a},\cdot)$). Furthermore, contrary to \cite{RST22}, it is not possible to reduce the sizes of the sets as much (by considering an exact dimension of the kernel for instance since it is non-invariant).

The only possibility to reduce the size of the sets is by considering the sets of zeros $\bm{F}^{0} = \{\bm{x} \in \Fq^{m+n} ~|~ \F(\bm{x}) = \bm{0} \}$ and $\bm{P}^{0} = \{\bm{x} \in \Fq^{m+n} ~|~ \Pcal(\bm{x}) = \bm{0} \}$. Since only one element of $\left( \Tilde{F}_a \times \Tilde{F}_b \right) \cap \bm{F}^{0}$ is sampled, only the size of the set $\left( \Tilde{P}_a \times \Tilde{P}_b \right) \cap \bm{P}^{0}$ is of interest. The cardinal of $\Tilde{P}_a \times \Tilde{P}_b$ is $q^{m+n}$, and when using only zeroes, it is of size at least $q^{m+n-k'}$ (because the set of zeroes for $p_1$ is at least $q^{n+m-1}$, and the Grassmann formula gives us the intersection with the set of zeroes of the quadratic forms).

As in the first algorithm, there are $(q-1)^2$ possibilities for $\left((\bm{a},\bm{b}),(\bm{a}',\bm{b}')\right)$ that gives a collision, allowing us to obtain the complexity given in Lemma \ref{lem:qmlealg2}. Finally, since $|\Tilde{F}_a| \times |\Tilde{F}_b|$ is equal to $\frac{M_{k,n}(n-1)q^m}{q^{kn}} \times \frac{M_{k,m}(m-1)q^n}{q^{km}} \approx q^{2m+2n-2k-2}$, if $2k\ge2n+2m$ then the algorithm will not work on average, due to the absence of elements in $\Tilde{F}_a \times \Tilde{F}_b$.

\begin{algorithm}[h]
\caption{A bilinear solver for the $\QSMLE$ problem}
\label{alg:QMLEalg2}
\begin{algorithmic}[1]
\REQUIRE An $\MSE$ instance $\bm{G} \in \Fq^{k \times mn},\bm{G}'\in \Fq^{k' \times mn}$. 
\ENSURE $\bm{A} \in \GL_m(q)$, $\bm{B} \in \GL_n(q)$.
\STATE Build $\Tilde{P}_b = \{ \bm{b} \in \Fq^{n} ~|~ \operatorname{dim}\operatorname{Ker} \Pcal(\cdot,\bm{b}) > 0  \}$.
\STATE Build $\Tilde{P}_a =\{\bm{a} \in \Fq^{m} ~|~ \operatorname{dim}\operatorname{Ker} \Pcal(\bm{a},\cdot) > 0  \} $ from $\Tilde{P}_b$.
\STATE Build $\Tilde{F}_b = \{ \bm{b} \in \Fq^{n} ~|~ \operatorname{dim}\operatorname{Ker} \F(\cdot,\bm{b}) > 0  \}$.
\STATE Build $\Tilde{F}_a =\{\bm{a} \in \Fq^{m} ~|~ \operatorname{dim}\operatorname{Ker} \F(\bm{a},\cdot) > 0  \} $ from $\Tilde{F}_b$.
\STATE Set $\bm{A} \leftarrow \bm{0}$ and $\bm{B} \leftarrow \bm{0}$.
\STATE $\bm{a} \sampler \Tilde{F}_{a} \times \Tilde{F}_{b}$.
\REPEAT
\STATE $\bm{b} \sampler \Tilde{P}_{a} \times \Tilde{P}_{b} \setminus E$.
\STATE $E = E \cup \{ \bm{b} \}$.
\STATE $(\bm{A},\bm{B}) \sampler \mathsf{Solve}_{\mathsf{\iQSMLE}}\left( \F(\bm{x}+\bm{a}), \Pcal(\bm{x}+\bm{b})\right)$. 
\UNTIL{$\mathsf{rank}\left(\begin{matrix}\bm{G} \\ \bm{G}'\big(\bm{A}^\top \otimes \bm{B}) \end{matrix} \right) = k$ and $\bm{A} \in \GL_m(q)$ and $\bm{B} \in \GL_n(q)$}
\STATE Return $\bm{A},\bm{B}$.
\end{algorithmic} 
\end{algorithm}

\begin{lemma}\label{lem:qmlealg2}
    When the sets $\Tilde{F}_a$, $\Tilde{F}_b$, $\Tilde{P}_a $, $\Tilde{P}_b$ are non-empty, the Algorithm \ref{alg:QMLEalg2} finds a solution to the $\QSMLE$ problem in time at least $$(q^{m+n-k'-2}-n_{\mathsf{sol}}) \mathsf{C}_{\iQSMLE} +  n_{\mathsf{sol}}\mathsf{C}_{\iQSMLE} (mn)^\omega $$ where $n_{\mathsf{sol}}$ is the number of queries for which $\mathsf{Solve}_{\mathsf{\iQSMLE}}\left( \F(\bm{x}+\bm{a}), \Pcal(\bm{x}+\bm{b}) \right)$ returns $(\bm{A},\bm{B})$ instead of $\perp$, and $\mathsf{C}_{\iQSMLE}$ is the cost of running the $\mathsf{Solve}_{\mathsf{\iQSMLE}}$ algorithm.
\end{lemma}

\paragraph{Solving particular $\iQSMLE$.} While \cite{RST22} adapted the matrix-pencil algorithm from Bouillaguet for $\MCE$, the differences in the case of a subcode need to be addressed.

Let $\Pcal$ and $\F$ the two inhomogeneous quadratic maps such that $\Pcal(\bm{x})=\F(\bm{xS})\bm{T}$ with $\bm{S} = \begin{bmatrix}
   \bm{A}  & \bm{0} \\
    \bm{0} & \bm{B}^\top 
\end{bmatrix}$, $\bm{A} \in \GL_m(q)$, $\bm{B} \in \GL_n(q)$. Then, one can write $\Pcal(\bm{x}\bm{S}^{-1}) = \F(\bm{x})\bm{T}$. One can then obtain the same equations as in \cite{RST22}.

Due to the construction of the polynomials $\{ p_i \}_{i \in \oneto{k'}}$ and $\{f_j\}_{j \in \oneto{k}}$, which are built by applying an affine change of variables to bilinear forms, they can be written as:

\begin{equation*} \label{eqiqsmle0}
    \begin{cases}
  p_i(\bm{x}) = (x_1,\dots,x_m)\bm{P}_i^{(2)}(x_{m+1},\dots,x_{m+n})+\bm{P}_i^{(1)}\bm{x}^\top +y_i\\
  f_j(\bm{x}) = (x_1,\dots,x_m)\bm{F}_j^{(2)}(x_{m+1},\dots,x_{m+n})+\bm{F}_j^{(1)}\bm{x}^\top +y'_j
    \end{cases}
\end{equation*}
for some matrices $\bm{P}_i^{2},\bm{F}_j^{2} \in \Fq^{m \times n}$, $\bm{P}_i^{(1)},\bm{F}_j^{(1)} \in \Fq^{1 \times (m+n)}$

\begin{equation*}
    \begin{cases}
  p'_i(\bm{x}) = (x_1,\dots,x_m)\bm{A}^{-1}\bm{P}_i^{(2)}\bm{B}^{-\top}(x_{m+1},\dots,x_{m+n})+\bm{P}_i^{(1)}\begin{bmatrix}
   \bm{A}^{-1}  & \bm{0} \\
    \bm{0} & \bm{B}^{-\top} 
\end{bmatrix}^\top\bm{x}^\top +y_i\\
  f_j(\bm{x}) = (x_1,\dots,x_m)\bm{F}_j^{(2)}(x_{m+1},\dots,x_{m+n})+\bm{F}_j^{(1)}\bm{x}^\top +y'_j
    \end{cases}
\end{equation*}

Then we have the equality $(p'_1,\dots,p'_{k'}) = (f_1,\dots,f_k)\bm{T}$.
  By identifying the degree $\ell$ parts, we deduce the equations $\mathcal{S}_\ell$ in the variables $\bm{A},\bm{B},\bm{T}$:
  
  \begin{align*}
      \begin{cases}
        \bm{A}^{-1}\bm{P}_i^{(2)}\bm{B}^{-\top} = \sum\limits_{j=1}^k t_{j,i}\bm{F}_j^{(2)}~~~~~~ & (\mathcal{S}_2)\\
        \bm{P}_i^{(1)}\begin{bmatrix}
   \bm{A}^{-1}  & \bm{0} \\
    \bm{0} & \bm{B}^{-\top} 
\end{bmatrix}^\top = \sum\limits_{j=1}^k t_{j,i} \bm{F}_j^{(1)} & (\mathcal{S}_1)\\
        y_i = \sum\limits_{j=1}^k t_{j,i} y'_j & (\mathcal{S}_0)
      \end{cases}
  \end{align*}
  
  Trivially, the $k'$ equations of degree $0$ are linear in the variables of $\bm{T}$ only, as $\bm{S}^{-1}$ does not influence this part. In $\mathcal{S}_1$, there are $k'(m+n)$ linear equations, as $\bm{a}$ and $\bm{b}$ have been chosen randomly when building the $\iQSMLE$ instance. By applying the same trick from \cite{RST22} to obtain a new collision, i.e, finding $\bm{x_0}$ and $\bm{x_0}'$ such that $\bm{P}^{(1)}_i\bm{x_0} - y =0$ and $\sum\limits_{j=1}^kt_{j,i}\bm{F}^{(1)}_k\bm{x_0}' - y' =0$, one may add $k'$ equations to $\mathcal{S}_1$. However, such an approach does not work since one would have to iterate through $q^{n+m-k'}$ such elements $\bm{x_0}$ which would increase the complexity by a factor too large to be worth it. In the next lines, we give a lower bound for solving the system made out of $\mathcal{S}_\ell$.
  
  Thanks to $\mathcal{S}_0$ and $\mathcal{S}_1$, $k'$ variables in $\bm{T}$ can be replaced by constants and $k'm$ variables in $\bm{A}^{-1}$ and $k'n$ variables in $\bm{B}^{-\top}$ can be replaced by variables in $\bm{T}$ in the equations $\mathcal{S}_2$. Then a lower bound for solving would be the complexity to linearize in degree 3 where we have $(m^2-k'm)(n^2-k'n)(k-1)k'$ monomials in tri-degree $(1,1,1)$ for $k'(n-k')(m^2-k'm) + (k'm-k')(n^2-k'n)+(n-k')(m-k')(k-1)k'$ equations, $(m-k')\binom{(k-1)k'+1}{2}$ monomials in tri-degree $(1,0,2)$ for $(m-k')k'^2(k-1)+k'^2(m-k')m$ equations and $(n-k')\binom{(k-1)k'+1}{2}$ monomials in tri-degree $(0,1,2)$ for $(n-k')k'^2(k-1)+k'^2(n-k')n$ equations. The order of variables in tri-degree is assumed to be $\bm{A}^{-1},\bm{B}^{-\top},\bm{T}$. For each tri-degree, no linearisation is possible. Thus we need to go at least in degree 3 to linearise ending up with the following complexity:

\begin{lemma}The $\iQSMLE$ instance obtained from either Algorithm \ref{alg:QMLEalg} or Algorithm \ref{alg:QMLEalg2} \new{is estimated to be solved} in time at least $$\mathsf{C}_{\iQSMLE} = \mathcal{O}\left(\Big(\binom{v+2}{3}-v_{-}\Big)^\omega\right), \text{~~~~} n,k \rightarrow \infty$$ where \begin{equation*}
\begin{cases} v = n^2+m^2+kk'-2k'(m+n) \\v_{-} = \binom{m(m-k')+1}{2}\left(k'(k-1)+n(n-k')\right)+\binom{n(n-k')+1}{2}\left(k'(k-1)+m(m-k')\right)
\end{cases}
\end{equation*}\end{lemma}

Overall, \old{this attack is the one that allows us to select our parameters, as the exponential factor we consider is much lower than for Leon's algorithm}.

\subsection{Other attacks through invariant}

Other attacks also exist on $\MCE$, which must be mentioned. The two attacks that we will mention use the fact that, similarly as what is explained in \cite{MEDS2,KQT24} and \cite{RS24}, the $\MSE$ can be translated as finding $\bm{A} \in \GL_m(q)$, $\bm{B} \in \GL_n(q)$, and $\bm{T} \in \Fq^{k \times k'}$ such that  $\C(\bm{x},\bm{y},\bm{T}\bm{z}) = \D(\bm{A}^{-\top}\bm{x},\bm{B}^{-1}\bm{y},\bm{z})$ for $$\C  \text{ : } \Fq^m \times \Fq^n \times \Fq^k \mapsto \Fq \text{ , } \C(\bm{x},\bm{y},\bm{z}) = \sum_{i,j,l}\bm{C}_{i,j}^{l}x_iy_jz_l $$ $$\D  \text{ : } \Fq^m \times \Fq^n \times \Fq^{k'} \mapsto \Fq \text{ , } \D(\bm{x},\bm{y},\bm{z}) = \sum_{i,j,l}\bm{D}_{i,j}^{l}x_iy_jz_l$$ where $\bm{C}^l$ (resp. $\bm{D}^l$) is the $l$th matrix in the basis of $\C$ (resp. $\D$). In \ref{subsec:algatk}, a matrix representing $\C(\bm{x},\bm{y},\bm{z})$ (resp. $\D$) is given by $Vec_{row}(\bm{G}_A)(\bm{x}\otimes \bm{y}\otimes \bm{z})^\top$ (resp. $Vec_{row}(\bm{G}'_A)(\bm{x}\otimes \bm{y}\otimes \bm{z})^\top$). The first attack, from \cite{KQT24}, uses a walk on a graph to obtain a path that allows to recover the isometry. The second one, from \cite{RS24}, uses a rare invariant in a trilinear map, called a triangle, and then recovers the matrices $\bm{A}$ and $\bm{B}$. We explain hereafter why these attacks do not work with $\MSE$. Furthermore, the adaptation of these attacks either does not work (for the first one) or with a very low probability (for the second one), as they rely on invariants that are not present in $\MSE$.

\paragraph{Using corank of bilinear forms.} The idea of the attack from \cite{KQT24} is to look at subsets  $\mathbb{P}_{\C,R}= \{ \bm{u} \in \mathbb{P}(\Fq^m) ~~\vert~~ \Rank\big(\C(\bm{u},\textbf{--},\textbf{--})\big) = R \}$ and $\mathbb{P}_{\D,R}= \{ \bm{u} \in \mathbb{P}(\Fq^m) ~~\vert~~  \Rank\big(\D(\bm{u},\textbf{--},\textbf{--})\big) = R \}$ for a fixed value $R$. Once the subsets are built, the algorithm applies an invariant function, by building points $\bm{u}$, $\bm{v}$ and $\bm{w}$ where $\Rank\big(\C(\bm{u},\textbf{--},\textbf{--})\big) = R$, $\Rank\big(\C(\textbf{--},\bm{v},\textbf{--})\big) = R$ and  $\Rank\big(\C(\textbf{--},\textbf{--},\bm{w})\big) = R$. Then the isometry is found by matching these points. 

Aside from the fact this attack is specifically designed for parameters such that $m=n=k$ (and hence cannot work here as we will take larger values of $k$), in $\MSE$ the ranks of the bilinear forms obtained are not the same, because $\Rank\big(\C(\bm{u},\textbf{--},\textbf{--})\big)$ is at most $\min(n,k)$ while $\Rank\big(\D(\bm{u},\textbf{--},\textbf{--})\big)$ is at most $\min(n,k')$. Because $k'$ is much smaller than $k$ (and than $n$), the rank of the bilinear forms obtained is not an invariant, meaning it is not possible to look directly for collisions in the subsets $\mathbb{P}_{\C,R}$ and $\mathbb{P}_{\D,R}$. Furthermore, as the algorithm does not look for full rank matrices (it looks for elements such that the co-rank of $\Rank\big(\D(\bm{u},\textbf{--},\textbf{--})\big) = R$ is $1$), the fact that the rank is not invariant prevents any adaptation of the attack.

\paragraph{Using triangles in trilinear maps.} Another attack described in \cite{RS24} uses the same equivalence of trilinear maps (also called $3$-Tensor Isomorphism problem). This attack uses an invariant called a triangle, which is a vector $(\bm{u},\bm{v},\bm{w})$ such that $\C(\bm{u},\bm{v},\textbf{--}) = 0$, $\C(\textbf{--},\bm{v},\bm{w}) = 0$, $\C(\bm{u},\textbf{--},\bm{w}) = 0$. By finding a triangle for the code $\C$, and then a triangle for $\D$, the attack then obtains the matrices $\bm{A},\bm{B}$, and $\bm{T}$ thanks to an algebraic modeling and a computation of a Gröbner basis.

As stated in \cite{RS24}, these triangles occur with probability $1/q$ in the trilinear maps, and the cost of the algorithm is dominated by the cost of finding a triangle. The case of the subcode equivalence problem is different however. In fact, while a triangle is an invariant in the case of $\MCE$, this is not the case for $\MSE$, making the adaptation of the attack possible much less often.

More precisely, let $(\bm{u},\bm{v},\bm{w})$ be a triangle of $\D(\bm{x},\bm{y},\bm{z})$.
Then, we have that $\D(\bm{u},\bm{v},\bm{e}_i)=0$ for $\{\bm{e}_i\}_{i \in \oneto{k'}}$ a basis of $\Fq^{k'}$. However, this does not imply that $\C(\bm{A}^{\top}\bm{x},\bm{B}\bm{y},\bm{e}'_j)=0$ for $\{\bm{e}_i\}_{i \in \oneto{k'}}$ a basis of $\Fq^{k}$: we will only have that $\C(\bm{A}^\top\bm{x},\bm{B}\bm{y},\bm{T}\bm{e}_i)=0$ for all $\bm{e}_i$.

Reciprocally, if $(\bm{u},\bm{v},\bm{w})$ is a triangle for $\C$, then, this can imply a triangle in $\D$ only if $\bm{w}$ is in the image of the application 
\begin{align*}
    \psi_{\bm{T}} \text{ : } \Fq^{k'} &\rightarrow \Fq^{k}\\
     \bm{z} &\mapsto \bm{T} \bm{z}
\end{align*}

This image is a vector space of dimension $k'$, while there are $q^k$ elements $\bm{w}$ possible in a triangle. As a result, when in $\MCE$ the weak keys appear in proportion $1/q$, they appear for $\MSE$ with probability $\frac{1}{q^{k-k'+1}}$ (if a triangle appears in $\C$, which it does with probability $1/q$, then it leads to a triangle in $\MSE$ with probability at most $\frac{1}{q^{k-k'}}$, as it is the probability that one out of $q^{k'}$ elements belongs in a set of size $q^{k}$). As a result, the attack does not work in practice when $k'$ is small, as weak keys will not appear.

\paragraph{Adapting the triangles} First we rewrite the definition of a triangle $(\bm{u},\bm{v},\bm{w})$ for $\D$ as $\D(\bm{u},\bm{v},\bm{e}_i)=\D(\bm{u},\bm{f}_j,\bm{w}) = \D(\bm{g}_\ell,\bm{v},\bm{w}) = 0$ for all $i,j,\ell$ where $\Fq^{k'}=\langle e_i \rangle,\Fq^m=\langle f_j \rangle,\Fq^n =\langle g_\ell\rangle$. This is equivalent to $\C(\bm{A}^\top\bm{u},\bm{B}\bm{v},\bm{T}\bm{e}_i)=\C(\bm{A}^{\top}\bm{u},\bm{B}\bm{f}_j,\bm{T}\bm{w}) = \C(\bm{A}^{\top}\bm{g}_\ell,\bm{B}\bm{v},\bm{T}\bm{w}) = 0$ for all $i,j,\ell$. Thus a triangle for $\D$ does not correspond to a triangle for $\C$ but to a \emph{pseudo-triangle} $(\bm{u},\bm{v},\bm{w},\bm{W}) \in \mathbb{P}^m(\Fq) \times \mathbb{P}^n(\Fq) \times \mathbb{P}^{k}(\Fq) \times \mathsf{Gr}^{k'}(\Fq^k)$: $\C(\bm{u},\bm{v},\bm{W})=\C(\bm{u},\Fq^n,\bm{w}) = \C(\Fq^m,\bm{v},\bm{w}) = 0$ with $\bm{w} \in \bm{W}$. The idea of invariants is to deduce $(\bm{A},\bm{B},\bm{T})$ from a collision. To find such collision, we need to find pseudo-triangles for $\C$ which can be done by an analogous method as for the triangles. If $\D$ has a triangle, we know there exists a $\bm{W} \in \mathsf{Gr}^{k'}(\Fq^k)$ generated by $\{\bm{w}_\ell \}_{\ell \in \oneto{k'}}$ such that solving this system

\begin{equation*}
    \begin{cases}
    \C(\bm{u},\bm{v},\bm{w_\ell}) = 0, \forall 1\leq \ell \leq k'\\
    \C(\bm{u},\bm{f}_j,\bm{w}) = 0, \forall 1\leq j \leq n\\
    \C(\bm{e}_i,\bm{v},\bm{w}) = 0, \forall 1\leq i \leq m\\
    \bm{w} \in \bm{W}
    \end{cases}
\end{equation*}

will give a pseudo-triangle. However, we need first to guess $\bm{W}= \langle \bm{w}_1,\dots,\bm{w}_{\ell} \rangle \in \mathsf{Gr}^{k'}(\Fq^k)$ and then solve the system. This simple step of guessing $\bm{W}$ has already a cost $\CG{k}{k'}$, making the adaptation already unpractical, even without considering the solving of the aforementioned system. One could try to solve directly the system \old{without specifying the $\bm{w}_\ell$ vectors, which would add $k'k$ variables, switch the linear equations from $\bm{w}\in \bm{W}$ by equations of degree $k'+1$ and the first quadratic equations by cubic equations.} \new{by parametrizing $\bm{W}$. Completing $\{\bm{w}\}$ in a full basis and guessing an information set would still add $k'(k-k')$ variables and increase the degree of all the equations by one.}

\begin{remark}
There is a reason all the attacks that use any invariant do not perform well on $\MSE$. In fact, all such attacks work using the same process: 
\begin{itemize}
    \item Select an invariant property $\mathbb{P}$ on elements of $\C$ and $\D$ (this can be the rank of matrices as in the Leon algorithm, the triangle property in the attack from \cite{RS24}, or any other property);
    \item Build elements of $\C$ and $\D$ that possesses this property;
    \item Find a collision between these lists of elements.
\end{itemize}
However, such an approach is bound to have an increased complexity in the case of $\MSE$. There are, in fact, two possible cases when trying this approach with $\MSE$. The first one is when the property $\mathbb{P}$ is an invariant (which is the case in Leon's algorithm for instance). In that case, it is obvious that the set $\mathbb{P}_{\C} = \{\bm{x} \in \C | \text{~}\bm{x}\text{~has the property~} \mathbb{P} \}$ is larger than the one from $\D$. Naively, we can expect $|\mathbb{P}_{\D}|$ to be equal to $ |\mathbb{P}_{\C}|/q^{k-k'}$ as $|\D| = |\C|/q^{k-k'}$. This (rough) estimation leads us to think that collisions are much harder to find. 

When the property is not invariant, it is not possible to look for collisions, as what happens with triangles. Furthermore, adapting the property requires some kind of enumeration, which increases the complexity of finding a possible match.
\end{remark}

\section{Parameters choice and signature sizes} \label{sec:param}
We give in Table \ref{tab:MSEparam} the parameters that we are going to use for our signature scheme. They are chosen according to the attacks described in Section \ref{sec:attaques}. 
For Leon's attack, we focus solely on the two-collision approach, since the one-collision attack relies on algebraic methods for which we do not yet have a precise formula. Nevertheless, the estimation of the number of equations and monomials from Section \ref{sec:attaques} indicates that the complexity will not threaten our parameter sets. We then give in Table \ref{tab:IMSEPparameterstotal} the sizes obtained for the signature from $\IMSEP$ from Section \ref{sec:protocol}. \old{In Table \ref{tab:MSEparameterstotal}, we illustrate the difference with $\MSE$ to highlight the need to use an inhomogeneous version of the problem: at level I, the difference is around 1 kB (the parameters are the same for $\MSE$, except the value $k$, which is taken as $k+k'$ compared to $\IMSEP$).} The parameters $(T_{\text{open}},w)$ for both TCitH and VOLEitH frameworks are the same as what is taken in \cite{BFG+}.

A final remark can be made to compare the size of the signature we obtain and the size of other MPCitH-based signature schemes. In particular, one may wonder why we obtain larger sizes than MinRank, Rank Syndrome Decoding, Syndrome Decoding, or even PKP (\cite{BFG+}, \cite{BBGK24}). The reason is that the witness is composed of two invertible matrices, which do not have a smaller representation than $(m^2+n^2)\log_2(q)$, which is quite large.

\begin{table}[h]
\renewcommand*{\arraystretch}{1.5}
\centering
\scalebox{1}{
\begin{tabular}{ccccccccc c c}
\toprule
\multirow{2}{*}{NIST Security level} & \multirow{2}{*}{Parameter set} & \multirow{2}{*}{$q$} & \multirow{2}{*}{$m$} & \multirow{2}{*}{$n$} & \multirow{2}{*}{$k$} & \multirow{2}{*}{$k'$} & \multirow{2}{*}{$|\pk|$} & \multicolumn{2}{c}{Complexity ($\log_2$)} \\
\cmidrule(lr){9-10}
 & & & & & & & & Leon & $\QSMLE$ \\
\midrule
\multirow{2}{*}{I} & IMSEP-Ia & 64 & 12 & 12 & 32 & 3 & 268 B & 230 & 156 \\
                   & IMSEP-Ib & 128 & 11 & 11 & 30 & 3 & 255 B & 237 & 160 \\
\midrule
III & IMSEP-III & 64 & 18 & 18 & 50 & 3 & 641 B & 340 & 235 \\
\midrule
V & IMSEP-V & 64 & 22 & 22 & 67 & 3 & 963 B & 413 & 286 \\ 
\bottomrule
\end{tabular}}
\vspace{1.5mm}
\caption{Parameter sets chosen according to the attacks described in Section \ref{sec:attaques} and according public key size in Bytes. Complexities of the attacks are a lower-bound, as the complexities of solving the polynomial systems is still unknown.}
\label{tab:MSEparam}
\end{table}
\begin{table}[h!]
\renewcommand*{\arraystretch}{1.5}
\centering
\scalebox{1}{
\begin{tabular}{cccccccc}
\toprule
Security & MSE\new{P} Parameters & Framework & Trade-off & $N$ & $\tau$ & $T_{\text{open}}$ & $|\mathsf{Sig}|$ \\
\midrule
\multirow{4}{*}{NIST I} & \multirow{4}{*}{\texttt{IMSEP-Ia}} & \multirow{2}{*}{TCitH} & Fast & 256 & 20 & 113 & 7 162 B\\
                        &                                  &                       & Short & 2048 & 12 & 111 & 5 014 B\\
\cmidrule(lr){3-8}
                        &                                  & \multirow{2}{*}{VOLEitH} & Fast & 256 & 16 & 102 & 6 292 B\\
                        &                                  &                         & Short & 2048 & 11 & 99 & 4 828 B\\
\midrule
\multirow{4}{*}{NIST I} & \multirow{4}{*}{\texttt{IMSEP-Ib}} & \multirow{2}{*}{TCitH} & Fast & 256 & 20 & 113 & 7 097 B\\
                        &                                  &                       & Short & 2048 & 12 & 111 & 4 975 B\\
\cmidrule(lr){3-8}
                        &                                  & \multirow{2}{*}{VOLEitH} & Fast & 256 & 16 & 102 & 6 234 B\\
                        &                                  &                         & Short & 2048 & 11 & 99 & 4 788 B\\
\midrule
\multirow{4}{*}{NIST III} & \multirow{4}{*}{\texttt{IMSEP-III}} & \multirow{2}{*}{TCitH} & Fast & 256 & 30 & 178 & 21 108 B\\
                          &                                  &                       & Short & 2048 & 18 & 174 & 14 316 B\\
\cmidrule(lr){3-8}
                          &                                  & \multirow{2}{*}{VOLEitH} & Fast & 256 & 24 & 176 & 18 438 B\\
                          &                                  &                         & Short & 2048 & 16 & 162 & 13 428 B\\
\midrule
\multirow{4}{*}{NIST V} & \multirow{4}{*}{\texttt{IMSEP-V}} & \multirow{2}{*}{TCitH} & Fast & 256 & 39 & 247 & 40 129 B\\
                        &                                  &                       & Short & 2048 & 25 & 245 & 28 543 B\\
\cmidrule(lr){3-8}
                        &                                  & \multirow{2}{*}{VOLEitH} & Fast & 256 & 32 & 247 & 35 576 B\\
                        &                                  &                         & Short & 2048 & 22 & 248 & 27 041 B\\
\bottomrule
\end{tabular}}
\vspace{1.5mm}
\caption{Resulting signature sizes for the new signature scheme based on $\IMSEP$.}
\label{tab:IMSEPparameterstotal}
\end{table}

\paragraph{Computational cost.} We also give a rough explanation of the computational cost of the signature scheme, in particular, the signing time. Two main parts influence the signature speed: the \textit{symmetric} part and the \textit{MPC simulation} part. For the symmetric part, \cite{BBM+} gives a timing of $3.2$ milliseconds for the whole signing process of \texttt{FAESTER-128s} with a signature size of around 4.6 kB, and $3.6$ milliseconds for \texttt{KuMQuat-$2^8$-L1s} for 2.5 kB of signature. As a result, we can be confident in the fact that the symmetric part will not be slower than these timings. For the MPC simulation part, one has to perform $m^2n^2$ operations in $\Fq$ for the Kronecker product, and then multiply the matrix of size $\Fq^{k' \times mn}$ with a matrix of size  $mn \times mn$. Finally, the multiplication of a matrix of size ${k' \times mn}$ with a matrix of size $mn \times mn-k$ must be done, along with $k'(mn+k') \cdot \rho $
linear combination over $\Fq$. We can compare this with other MPC protocols: for \texttt{KuMQuat-$2^8$-L1s}, there are $\mathcal{O}(mn^2)$  multiplications in $\Fq$, with $q=2$ and $m=n=152$ compared to $\mathcal{O}(k'(mn)^2)$ multiplications in $\Fq$ for us, with $q=64$ and $m=n=12$. This results in around the same amount of multiplications in $\Fq$, but for $q=2$, this protocol would need to be repeated more times than for $q=64$ (to have an appropriate soundness error), so we can expect the size of $\Fq$ to have a limited influence. Furthermore, since parameters are taken with $N=2048$, the symmetric part tends to often be more important than the MPC emulation one. We can thus expect our scheme to have a signing and verification time in the same order of magnitude as other MPCitH schemes, which are known to be quite efficient with the recent frameworks.

\subsubsection*{Acknowledgments}
This research was partially funded by the French Agence Nationale de
la Recherche and plan France 2030 program under grant ANR-22-PETQ-0008
PQ-TLS.

\bibliography{ref}

\appendix

\section{Threshold-Computation-in-the-Head and VOLE-in-the-Head} \label{app:tcith}
\new{We now give a more thorough description of the TCitH and VOLEitH frameworks (although we refer to \cite{FR23} and \cite{BBd+} for full details).}

\subsection{Threshold-Computation-in-the-Head}

\new{As explained in Section \ref{sec:protocol}, this framework uses Shamir's secret sharing scheme instead of additive sharing. The main point of using such a sharing scheme is to allow the multiplication of the shares in the MPC protocol. To do so efficiently, \cite{FR23} introduced a clever way to generate shares, by using \textit{replicated} Shamir's secret sharing, i.e., using additive sharing inside the Shamir's sharing (see \cite{FR23} for a full explanation).}

\new{We detail the functioning of this framework, when combined with the MPC protocol for polynomial constraints checking given in Section \ref{sec:protocol}, i.e., we are given $m$ polynomial $f_1,\dots,f_m \in \Fq[x_1,\dots,x_n]$ such that $f_1(\bm{x}) = \dots = f_m(\bm{x}) = 0$ for a witness $\bm{x}$.}

\new{Recall that the protocol evaluates $f_j(\share{\bm{x}}_i)$ for all $j$ in $\oneto{m}$, and then computes a linear combination $\share{\alpha}_i= \sum_{j=1}^{m}\gamma_jf_j(\share{\bm{x}}_i)$. The obtained value $\alpha$ must be equal to $0$ if the witness is valid.}

\new{This MPC protocol is then turned into a proof of knowledge by applying the TCitH framework, which follows the same principles as the MPC-in-the-Head one:}
\new{
\begin{itemize}
    \item The prover shares his witness $\bm{x}$ in $N$ shares using Shamir's Secret Sharing, with a polynomial of degree one. He commits to these $N$ shares and sends the commitments to the verifier;
    \item The verifier sends the challenge $(\gamma_1,\dots,\gamma_m) \in \Fq^m$ to the prover (actually, he sends $\rho$ vectors $\bm{\gamma}$);
    \item The prover performs the polynomial constraints checking protocol "in his head" using the challenge and sends a commitment of the shares of $\alpha$;
    \item The verifier asks for one share of $\bm{x}$, as only two of them are necessary to retrieve the witness;
    \item The prover sends one share of $\bm{x}$, as well as $d-1$ additional share of $\alpha$. From them, the verifier can check the consistency of the commitments and of the computations (only $d-1$ additional shares of $\alpha$ are needed, as the verifier already knows that $\alpha=0$ if the prover is honest). \end{itemize}
}
\new{This protocol has a soundness error of $$\frac{1}{q^\rho}+\big(1-\frac{1}{q^\rho} \big) \cdot \frac{d}{N}.$$ In fact, the false-positive probability of the polynomial constraint checking protocol is $\frac{1}{q}$ when each $\gamma_i$ is in $\Fq$. The idea of the proof of this probability is that when a prover does not know the witness, there is at least one $f_i$ such that $f_i(\bm{x}) \ne 0$ for the value of $\bm{x}$ the prover uses. Then, since all $\gamma_j$ are chosen uniformly random, the false-positive probability is easy to obtain if only one $f_i(\bm{x})$. If several $f_i(\bm{x}) \ne 0$, the linear combination is uniformly random, as, once again, all the $\gamma_j$ are uniformly random, which gives the $\frac{1}{q}$ probability. The MPC protocol is repeated $\rho$ times to reduce this probability, which brings the soundness of the TCitH scheme to the given value \cite[Theorem 2]{FR23}. }

\new{ The shares of the witness are generated using a pseudorandom generator as described in \cite[Section 3]{FR23}, which allows the prover to transmit the shares requested.  
Finally, the proof of knowledge is transformed into a signature thanks to the well-known Fiat-Shamir transform: the protocol is repeated $\tau$ times, and the verifier is removed by sampling the challenges using the hash of the commitment. We do not describe in detail the proof of knowledge and the signature scheme, as the framework is very generic. Instead, we invite the reader to refer to \cite{FR23}, in particular for the proofs of soundness and zero-knowledge. The number of repetitions $\tau$ is chosen such that $\left( \frac{d}{N}\right)^\tau \le 2^{-\lambda}$, and $\rho$ such that $\frac{1}{q^\rho} \le 2^{-\lambda}$.}

\new{As to the size of the signature, it consists in:
\begin{itemize}
    \item A value $\Delta \bm{x}$, of size $|\bm{x}|$;
    \item The shares $\share{\alpha}_J$, for a set $J$ of size $d-1$, for the $\rho$ repetitions of the MPC protocol;
    \item A commitment of size $2\lambda$;
    \item A sibling path to reveal the seeds used in the construction of the shares, which is of size $\lambda \cdot \log_2(N)$.
\end{itemize}
which results in a size of $$\text{Size}_{\text{TCitH}} = 4\lambda + \tau \left( \underbrace{|\bm{x}|}_{\Delta\bm{x}}+ \underbrace{(d-1) \rho  \log_2(q)}_{\share{\alpha}_J}  + \underbrace{2\lambda}_{\mathsf{com}_I}+ \underbrace{\lambda \cdot \log_2(N)}_{\text{GGM Tree}} \right) $$}

\subsection{The VOLEitH framework}
\new{
The VOLEitH framework has been introduced in \cite{BBd+}, and uses VOLE correlations. While this framework appeared before the TCitH one, \cite{FR23} showed that it can be expressed as an instantiation of the TCitH framework, by taking $\ell = 1$ and hiding the secret as the coefficient of degree $1$ in the polynomial. More precisely, while the TCitH framework shares a secret $\bm{x}$ with randomness $r$ as $P(X) = r\cdot X+\bm{x}$, the VOLEitH framework shares it as $P(X) = \bm{x} \cdot X + r$. With this construction, the same operations are possible, the only difference being the addition between a share $\share{\bm{x}_1}$ of degree $d_1$ and $\share{\bm{x}_2}$ of degree $d_2$ where $d_1 \ge d_2$, which requires to multiply $\share{\bm{x}_2}$ by $X^{d_1-d_2}$ (where $X$ is replaced by the evaluation point used in the sharing).}\new{Furthermore, the VOLEitH framework benefits an optimization: after committing to the $\tau \cdot N$ shares of $\bm{x}$, $\share{\bm{x}}_i^{j}$, $i \in \oneto{N}$ and $j \in \oneto{\tau}$, the prover can merge sets of $\tau$ sharings to lower the soundness error. Concretely, given a set $\share{\bm{x}}_{i_1}^{(1)},\dots,\share{\bm{x}}_{i_\tau}^{\tau}$, the prover can use a morphism $\phi$ from $\Fq^{\tau}$ to $\mathbb{F}_{q^\rho}$ where $q^\tau \ge 2^\lambda $ and $\rho \ge \tau$ to obtain a share $\share{\bm{x}}_{1+i_1 \cdot N+\dots + i_\tau \cdot N^{\tau-1}}$=$ \phi\left(({\share{\bm{x}}_{i_e}}^{(e)})_{e \in \oneto{\tau}}\right)$.  As a result, $\tau$ shares of the same value can be seen as only one share in a larger field (see \cite{FR23} and \cite{BBd+}), which reduces the soundness error to $\frac{d}{N^\tau}$. For this modification to work however, the $\tau$ shares that are used must be shares of the same value $\bm{x}$. This introduces a parameter, $B$, chosen such that $B\cdot q \ge 16$, and a correction vector $\bm{u}$ for $\tau-1$ rounds. We will use this construction in a "black-box" way to build our signature scheme, the exact same as what we do with TCitH. We refer to \cite{BBd+} and \cite{FR23} for more details on the framework and its relation to TCitH.}
\new{
Finally, when transforming this PoK into a signature, the number of repetitions $\tau$ is chosen such that $\frac{d}{N^\tau} \le 2^{-\lambda}$, and $\rho$ such that $\frac{1}{q^\rho} \le 2^{-\lambda}$. Overall, this gives a signature with the following size: \begin{align*}\text{Size}_{\text{VOLEitH}} =  &4\lambda + (\tau-1) \left( \underbrace{|\bm{x}| + \rho \log_2(q) + (\rho+B) \log_2(q)}_{\text{Corrections of the shares}}\right)  + \underbrace{(\rho + B) \log_2(q)}_{\share{\alpha'}_J}\\& + \tau\cdot \left( \underbrace{\lambda \cdot \log_2(N)}_{\text{GGM Tree}} +  \underbrace{2\lambda}_{\mathsf{com}_I}\right) + \underbrace{|\bm{x}|}_{\Delta\bm{x}}+ \underbrace{(d-1)\rho \log_2(q)}_{\share{\alpha}_J} \end{align*}}
\new{Both frameworks also benefit from the optimization on GGM-trees from \cite{BBM+}, which reduces the cost of revealing the shares, replacing the value usual value of $\tau \cdot \lambda\cdot  \log_2(N)$ by $\lambda\cdot  T_{\text{open}}$. The optimization on GGM-trees also introduces proofs of work in the Fiat-Shamir transform, which allows the signer to use a parameter $w$, which corresponds to the number of zeros one must have in the last hash digest. The signer can then choose $\tau$ such that the security of the transform is only $\lambda - w$. Essentially, this optimization allows to have one less round in the Fiat-Shamir transform, due to having the additional margin of security $w$.}
\end{document}